\documentclass[journal]{IEEEtran}
\usepackage{bbm}
\usepackage{caption}
\usepackage{subfigure}
\usepackage{graphicx}
\usepackage{latexsym}
\usepackage{diagbox}
\usepackage{changepage}
\usepackage[fleqn]{amsmath}
\usepackage{amsmath}
\usepackage{amsfonts}
\usepackage{indentfirst}
\usepackage{CJK}
\usepackage{slashbox}
\usepackage{balance}
\usepackage{indentfirst}% Some very useful LaTeX packages include:
\usepackage[varg]{txfonts}% (uncomment the ones you want to load)
\usepackage{stfloats}% *** MISC UTILITY PACKAGES ***
\usepackage{multirow}%
\usepackage{booktabs}%\usepackage{ifpdf}
\usepackage{color,soul}% Heiko Oberdiek's ifpdf.sty is very useful if you need conditional
\usepackage{epstopdf}% compilation based on whether the output is pdf or dvi.
\usepackage{float}% usage:
\usepackage{bm}
\usepackage{cite}
\usepackage{makecell}
\usepackage{array}
\usepackage{bm}
\usepackage{mathtools}
\usepackage{geometry}
\usepackage[linesnumbered,ruled,commentsnumbered,longend]{algorithm2e}
\makeatletter

\newcommand{\Rmnum}[1]{\expandafter\@slowromancap\romannumeral #1@}

\makeatother
\makeatletter
\newtheorem{theorem}{Theorem}

\newenvironment{proof}[1][Proof]{\begin{trivlist}
		\item[\hskip \labelsep {\itshape #1}]}{\end{trivlist}}

\newenvironment{remark}[1][Remark]{\begin{trivlist}
		\item[\hskip \labelsep {\bfseries #1}]}{\end{trivlist}}

\newcommand{\qed}{\nobreak \ifvmode \relax \else
	\ifdim\lastskip<1.5em \hskip-\lastskip
	\hskip1.5em plus0em minus0.5em \fi \nobreak
	\vrule height0.75em width0.5em depth0.25em\fi}

\SetAlgoLongEnd

\ifCLASSINFOpdf
\else
\fi

\begin{document}

\title{Codebook-Based Max-Min Energy-Efficient Resource Allocation for Uplink mmWave MIMO-NOMA Systems}
\author{Wanming Hao,~\IEEEmembership{Member,~IEEE,} Ming Zeng,  Gangcan Sun, Osamu Muta,~\IEEEmembership{Member,~IEEE,} Octavia A. Dobre,~\IEEEmembership{Senior Member,~IEEE},  Shouyi Yang, Haris Gacanin,~\IEEEmembership{Senior Member,~IEEE}
	\thanks{This work was supported by Henan Post-Doctoral Science Foundation under Grant 19030015, and the National Natural Science Foundation of China under Grant U1604159. The work of Ming Zeng and Octavia A. Dobre have been supported by the Natural Sciences and Engineering Research Council of Canada (NSERC), though its Discovery program.}% \textit{(Corresponding author: Gangcan Sun.)}}
	\thanks{W. Hao, G. Sun, and S. Yang are with the School of Information Engineering, Zhengzhou University, Zhengzhou 450001, China, and W. Hao is also with the 5G Innovation Centre, Institute of Communication Systems, University of Surrey, Guildford GU2 7XH, U.K. (E-mail: \{iewmhao, iegcsun, iesyyang\}@zzu.edu.cn)}
	\thanks{M. Zeng and O. A. Dobre are with the Faculty of Engineering and Applied Science, Memorial University, St. Johns, NL A1B 3X9, Canada. (E-mail: \{mzeng, odobre\}@mun.ca)}
	\thanks{O. Muta is with Center for Japan-Egypt Cooperation in Science and Technology, Kyushu University, Fukuoka 819-0395, Japan.  (E-mails: muta@\{ieee.org, ait.kyushu-u.ac.jp\})}
		\thanks{H. Gacanin is with Nokia Bell Labs, Antwerpen, Flanders, Belgium. (E-mail:harisg@ieee.org)}
	}

\maketitle
\begin{abstract}
In this paper, we investigate the energy-efficient resource allocation problem in an uplink non-orthogonal multiple access (NOMA) millimeter wave system, where the fully-connected-based sparse radio frequency chain antenna structure is applied at the base station (BS).  To relieve the pilot overhead for channel estimation, we propose a codebook-based analog beam design scheme,  which only requires to obtain the equivalent channel gain.  On this basis, users belonging to the same analog beam are served via NOMA. Meanwhile, an advanced NOMA decoding scheme is proposed by exploiting the global information available at the BS.  Under predefined minimum
rate and maximum transmit power constraints for each user, we formulate a max-min user energy efficiency (EE) optimization problem by jointly optimizing the detection matrix at the BS and transmit power at the users. 
We first transform the original fractional objective function into a subtractive one. Then, we propose a two-loop iterative algorithm to solve the reformulated problem. Specifically, the inner loop updates the detection matrix and transmit power iteratively, while the outer loop adopts the bi-section method.  Meanwhile, to decrease the complexity of the inner loop, we propose a zero-forcing (ZF)-based iterative algorithm, where the detection matrix is designed via the ZF technique. Finally, simulation results show that the proposed schemes obtain a better performance in terms of spectral efficiency and EE than the conventional schemes.
\end{abstract}


% \fi
\IEEEpeerreviewmaketitle

\section{Introduction}
Millimeter wave (mmWave) technology has become a promising solution to satisfy the rapidly increasing capacity requirement in wireless networks. However, mmWave signals suffer from severe propagation loss due to high carrier frequency. To compensate the loss, a large number of antennas are usually employed at the base station (BS) to provide a large array gain~\cite{Rangan_Proc_2014}. Nonetheless, it is not practical to implement dedicated radio frequency (RF) chains for all antenna elements due to the high power consumption. For example, the power consumption of each RF chain at mmWave frequency can go up to 250 mW, which is five times larger than that in microwave frequency~\cite{Heath_JSTSP_2016,Dai_CMAG_2018}. Therefore, to reduce the energy consumption and hardware cost, advanced sparse RF chain  antenna structures have been applied, i.e., the number of RF chains is much lower than that of antennas. For example, fully-connected and~subarray structures are proposed in~\cite{Heath_TWC_2015}, while a lens-antenna array structure is investigated~\cite{Dai_JSAC_2017,Hao_TWC_2018}. 

{To achieve higher spectral efficiency (SE) for the wireless networks, non-orthogonal multiple access (NOMA) technology has been introduced as a promising solution. In general, the existing NOMA schemes can be classified into two categories: power-domain NOMA~\cite{Zeng_WC_2018} and code-domain NOMA ~\cite{Chi_TWC_2018,Chi_TVT_2019}. In power-domain NOMA, multiple users transmit their signals sharing the same time-frequency-code resources, while user signals are differentiated in power domain~\cite{Dobre_CST_2017, Liu_2017_JSAC}. Code-domain NOMA is similar to code division multiple access, and the major difference is that low-density sequences and/or  sparse code multiple access are used in the former~\cite{Dobre_VTM_2018}. In this paper, we focus on the power-domain NOMA technique.} By combining mmWave multiple-input multiple-output (MIMO) and NOMA technologies, a mmWave MIMO-NOMA system is formed, which  represents an effective scheme to satisfy the high capacity and service quality demands of the wireless networks.
	
There are two major challenges in mmWave MIMO-NOMA systems. The first one is related to user clustering, i.e, how to divide the users to form NOMA clusters. So far, most clustering schemes are designed under the assumption of perfect knowledge of the channel state information (CSI)~\cite{Zeng_JSAC_2017,Hao_WCL_2017,Zeng_WCL_2017,Guo_TVT_2019}. Specifically, users with highly correlated channels are grouped together. Although the above clustering approach achieves a good system performance, obtaining the perfect  CSI of all users is not practical. In particular,  when the BS is equipped with a large number of antennas, the pilot overhead is huge~\cite{Lu_JSAC_2014}. The other challenge is related to the energy efficiency (EE) optimization on the uplink, an important metric  to evaluate the system performance. {Unlike the works focusing on the downlink EE~\cite{Hao_WCL_2017,Dai_JSAC_2017,Hao_IoT_2019},  the optimization of uplink EE is more  challenging since the power allocation at users and beam design at the BS must be taken into consideration jointly. {Although joint optimization problems in the downlink exist in the  literature, such as  joint subcarrier and BS power allocation~\cite{Fu_TVT_2018,Li_TVT_2017} as well as joint user  access and BS power allocation~\cite{Moltafet_TCOM_2018, Fang_JSAC_2017}. However, these optimization problems are clearly different from our joint optimization problem of the detection matrix at the BS and  transmit power at the users.  On the other hand, for downlink MIMO-NOMA, the variables are often the beamforming matrix, and semidefinite programming (SDP) is usually adopted for solving such problems. In contrast, for our considered problem, both  detection matrix at the BS and power values at the users need to be optimized. Moreover, these two different variables are coupled in the problem formulation, since the signal-to-interference-plus-noise ratio  (SINR) of the users is a function of their multiplier. Because of this, SDP may no longer be used. Besides, in downlink, there is only a total power constraint, while in uplink, each user has its own power constraint. As a result, the existing solutions for downlink cannot be used to solve our considered problem. In this paper, we investigate an uplink EE optimization problem in the mmWave MIMO-NOMA  systems. The main contributions are summarized as follows:}

\begin{itemize}
	\item {{We design a user clustering scheme with analog beam alignment, where the BS sends the analog precoded reference signals (analog beams) generated by a given codebook in downlink. Users feed the received signal  strength and estimated CSI from all beam directions back to the BS, and the BS decides the optimal analog beams and user clustering according to the received information. After that, users belonging to the same cluster are served with NOMA.
	\item  An advanced NOMA decoding scheme is proposed by exploiting global information at the BS. To ensure user fairness, we consider maximizing the minimum user EE subject to limited transmit power and minimum rate requirements for each user. However, the above optimization problem is non-convex due to the non-convexity and non-smoothness of the fractional objective function, which is difficult to solve directly.
	\item  We transform the fractional objective function  into a subtractive one, and a two-loop iterative algorithm is proposed to solve the formulated problem.} For the inner loop, the transmit power and detection matrix are alternatively updated. Specifically, we first fix the transmit power and update the detection matrix, and then update the transmit power based on the obtained detection matrix.  The above iteration is carried out till convergence. For the outer loop, the classical bi-section algorithm is adopted.
	\item Finally, to reduce the complexity for the inner loop, we apply the zero-forcing (ZF) technique to obtain the detection matrix, and then only optimize the transmit power to  solve the formulated problem. Meanwhile, simulations are conducted, which show that the proposed schemes achieve a higher SE and EE in comparison  with the conventional schemes.}  
\end{itemize}

The remainder of this paper is organized as follows. In Section \Rmnum {2}, the related works are summarized. In Section \Rmnum {3}, the system model and max-min EE-optimal problem are introduced. In Section  \Rmnum {4}, the proposed two-loop iterative algorithm is presented. The ZF-based low-complexity algorithm is proposed in Section \Rmnum {5}. Numerical results are given in Section~\Rmnum {6}, and conclusions are drawn in Section \Rmnum {7}.

The notations in this paper are as follows: $(\cdot)^T$ and $(\cdot)^H$ represent the transpose and Hermitian transpose, respectively, $\|\cdot\|$ means the Euclidean norm, $\mathbb{E}\{\cdot\}$ means the expectation operator. {Re($\cdot$)} denotes the real number operation. The key acronyms are summarized in Table~\ref{table1}.
\section{Related Works}

{Currently, research on MIMO-NOMA has gained significant attention. \cite{N_TCOM_2019} introduces the NOMA transmission at an unmanned aerial vehicle BS for serving more users simultaneously. Then, a beam scanning approach is proposed to maximize the sum rate of the system. In~\cite{Lv_IoT_2018}, the authors assume that multiple machine-type communication (MTC) devices share the same communication resources in the mmWave-NOMA system, and propose an MTC pairing scheme based on the distance between the BS and the MTC devices. Finally, closed-form expressions of outage probability and sum rate are derived. A low-complexity iterative linear minimum mean square error (LMMSE) multiuser detector is proposed for the MIMO-NOMA system~\cite{Liu_TSP_2019}, and the authors prove that the proposed matched iterative LMMSE detector can achieve optimal capacity for any number of users. In~\cite{Mao_2018_EUR}, the authors compare the three multiple access schemes, including space-division multiple access, rate-splitting multiple access (RSMA) and NOMA, and their results show that the RSMA scheme can obtain a higher performance gain in comparison with the other two schemes under certain conditions. However, the above works mainly focus on the SE without considering the EE.}
	
For the EE problem in MIMO-NOMA system, the authors in~\cite{Dai_JSAC_2017} propose a NOMA scheme according to the formed beamspace, where users selecting the same beam are grouped into the same cluster. {Note that our investigated problem is totally different from~\cite{Dai_JSAC_2017} in the following two aspects: {{\textit{i)}} We investigate the user fairness-based EE maximization problem, while~\cite{Dai_JSAC_2017} considers the SE maximization problem. Note that EE in~\cite{Dai_JSAC_2017} is simply defined as the maximum sum rate over the corresponding consuming power, and thus, the EE maximization is actually the SE maximization; {\textit{ii)}} We design an optimization scheme for joint detection matrix at the BS and power allocation at the users.}  In contrast,~\cite{Dai_JSAC_2017} adopts the ZF precoding scheme and transforms the original problem into a power allocation one.} In~\cite{Hao_WCL_2017}, the authors propose a hybrid analog/digital precoding and power allocation scheme to maximize the EE of the system. Users are grouped into multiple NOMA clusters according to the channel correlation, and digital precoding design depends on the ZF technology to partially cancel the inter-cluster interference. {The authors in~\cite{Rahmati_ICC_2019} design the RSMA  and NOMA scheme in a cellular connected UAV network, and then investigate the EE of two schemes with mmWave.}

However, above works all focus on the downlink EE, and they mainly consider the  analog/digital precoding design at the BS. In this paper, we consider an uplink MIMO-NOMA system. In fact, the research of uplink EE is more meaningful because the user terminals are power-constrained. Furthermore, it is more challenging because both power optimization at user terminals and  beam design at the BS are needed. 
Although our previous work~\cite{Zeng_TVT_2019} considers the EE maximization problem in an uplink MIMO-NOMA mmWave network, the design of the detection matrix at the BS only depends on the ZF technique, which~limits the performance of the system. Furthermore, user fairness is not considered and the users clustering also depends on full CSI. {To summarize, compared with the previous works, the main contributions of this paper include: ({\textit{i}}) design an effective analog beam alignment-based user clustering scheme; ({\textit{ii}}) propose an advanced NOMA decoding scheme; ({\textit{iii}}) jointly optimize the users transmit power and the BS detection matrix  to maximize the minimum user EE such that user fairness is ensured.}

\begin{figure}[t]
	\begin{center}
		\includegraphics[width=9cm,height=4cm]{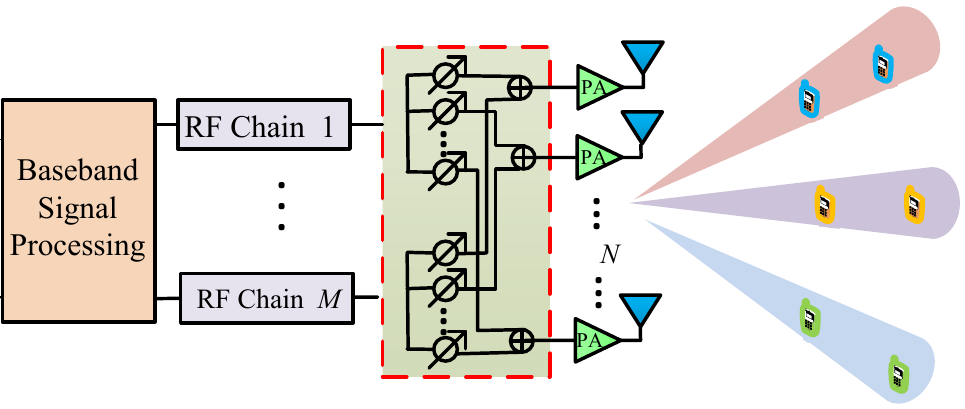}
		\caption{The uplink mmWave MIMO-NOMA model with sparse RF chain antennas structure.}
		\label{figure1}
	\end{center}
\end{figure}

\section{System Model and Problem Formulation}
In this section, we first describe the system model and beam alignment-based user clustering scheme. Then, an advanced NOMA decoding scheme is proposed by exploiting the global information available at the BS. Finally, we formulate the max-min uplink EE optimization problem.
\begin{table}[t]
	\caption{{Summary of Key Acronyms.}}
	\label{table1}
	\begin{center}
		\begin{tabular}{|c|c|}
			\hline
			Acronyms & Descriptions \\\hline
			MIMO&Multiple-input multiple-output \\\hline
			NOMA&Non-orthogonal multiple access\\\hline
			SE&Spectral efficiency\\\hline
			EE &Energy efficiency \\\hline
			BS & Base station \\\hline
			ZF & Zero-forcing\\\hline
			SIC & Successive interference cancellation\\\hline
			CSI&Channel state information\\\hline
			QoS& Quality-of-service\\\hline
			MmWave&Millimeter wave\\\hline
			RF&Radio frequency \\\hline
			LMMSE&Linear minimum mean square error\\\hline
			DFT&Discrete fourier transform\\\hline
			SINR&Signal-to-interference-plus-noise ratio\\\hline
			SNR&Signal-to-noise ratio\\\hline
			DC&Different of convex \\\hline
			CCCP&Constrained concave convex procedure\\\hline
			OMA&Orthogonal multiple access\\\hline
			TDMA&Time division multiple access\\\hline
			CSMA&Rate-splitting multiple access\\\hline
			MTC&Machine-type communication\\\hline
			%UAV&Unmanned aerial vehicle\\\hline
			SDP&Semidefinite programming\\\hline 
		\end{tabular}
	\end{center}
\end{table} 

\subsection{System Model and User Clustering}
We consider an uplink mmWave system as shown in Fig.~\ref{figure1}, where the BS is equipped with $N$ antennas and $M$ $(M\leq N)$ RF chains. Here, the fully-connected structure is considered, namely each RF chain is connected to all antennas through $N$ phase shifters, while each user is equipped with single antenna. We assume that the analog beam matrix is selected from a predefined codebook. To obtain high antenna gain and low beam gain loss, we adopt the  discrete Fourier transform (DFT) codebook, which is defined as~\cite{Zhou_PIMRC_2012}
\begin{eqnarray}
	{\bf{F}}(n,k)=\frac{1}{\sqrt{N}}e^{\frac{j2\pi(k-1)(n-1)}{K}}, \forall n\in \mathcal{N}, \forall k\in \mathcal{K},
\end{eqnarray}
where $\mathcal{N}\in\{1,\dots, N\}$ denotes the set of BS antennas, and $\mathcal{K}\in\{1,\dots,K\}$ is the set of beam patterns in codebook. Therefore, the DFT codebook $\bf{F}$ is a $N\times K$ matrix, and we show its polar plot in Fig.~\ref{figure2}. 

{Generally, in MIMO-NOMA systems, users are first grouped into clusters, and then, NOMA is applied among users within the same cluster~\cite{Lv_IoT_2019,Zhou_TWC}. In mmWave systems, users within one cluster typically share one analog beamformer since the number of RF chains is smaller than that of the transmit antennas~\cite{Dai_JSAC_2017}. Therefore, how to allocate the analog beamformer to users is important and challenging, especially when perfect CSI is unavailable.   Beam alignment is originally used for channel estimation via beam-scan procedure  in mmWave systems~\cite{Std_2009}. In this paper, we propose to apply it for user clustering. Meanwhile, the effective CSI can also be obtained with beam alignment.  The specific scheme and its rationality are described as follows: the BS first sends reference signals from each direction defined in codebook $\bf{F}$. Next, all users measure the received signal strength and estimate the effective channel at each beam direction. Here, each column of $\bf{F}$ stands for a beam direction, namely ${\bf{F}}=[{\bf{f}}_1,\dots,{\bf{f}}_K]$. To this end, the effective channel at beam direction ${\bf{f}}_k$ can be expressed as ${\bf{f}}_k^H{\bf{h}}_i$, where ${\bf{h}}_i\in\mathbb{C}^{N\times 1}$ denotes the channel gain between $N$ antennas at the BS and the $i$th user. After that, users feed the above results to the BS, including the received signal strength and estimated effective CSI. Meanwhile, the BS  decides the appropriate analog beam for each user based on the signal strength provided by the users.} Note that it is very likely that one analog beam is used to serve multiple users (especially for ultra-dense user distribution), which means that those users receive the strongest signal from the same beam, as shown in Fig.~\ref{figure1}. In this case, NOMA is employed among those users to improve the SE of the system.  In this paper, we assume that each analog beam can serve at least two users. Nonetheless, to decrease the decoding complexity, we only consider the two-user case, which is also the standard implementation for NOMA in Release 13 of the 3GPP. When more users are located in one beam coverage area, a proper user pair can be selected according to their channel gain difference, as in~\cite{Hao_WCL_2017}.  Note that our proposed scheme can be directly extended to an arbitrary number of users.  

\begin{figure}[t]
	\begin{center}
		\includegraphics[width=5cm,height=5cm]{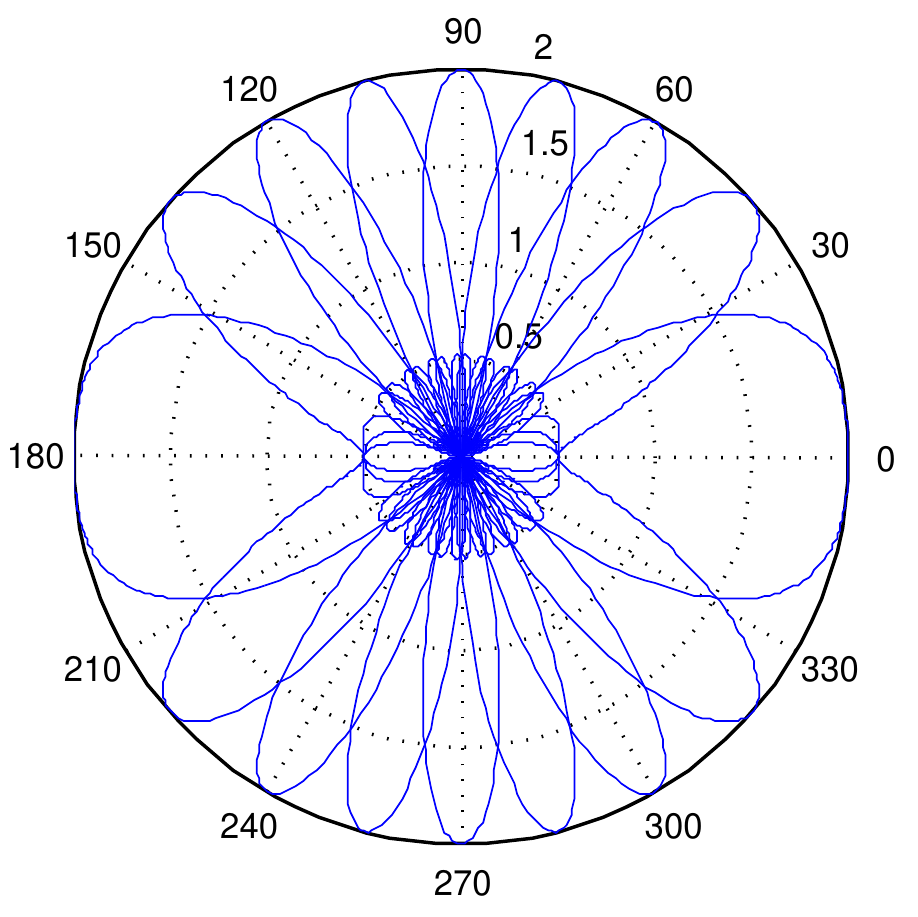}
		\caption{Polar plot for array factor of codebooks with $N=4$, $K=8$.}
		\label{figure2}
	\end{center}
\end{figure}

In this paper, we assume there are $M$ RF chains, which means that we can select $M$ directional beams from the DFT codebook. Towards lowering the inter-beam interference, the interval among selected beams is designed  as far as possible. For example, when there are 16 beams and 4 RF chains, namely ${\bf{F}}=[{\bf{f}}_1,{\bf{f}}_2,\dots,{\bf{f}}_{16}]$, we can form the analog beam at the BS as ${\bf{W}}_1=[{\bf{f}}_1,{\bf{f}}_5,{\bf{f}}_9,{\bf{f}}_{13}]^T$, ${\bf{W}}_2=[{\bf{f}}_2,{\bf{f}}_6,{\bf{f}}_{10},{\bf{f}}_{14}]^T$, ${\bf{W}}_3=[{\bf{f}}_3,{\bf{f}}_7,{\bf{f}}_{11},{\bf{f}}_{15}]^T$, and ${\bf{W}}_4=[{\bf{f}}_4,{\bf{f}}_8,{\bf{f}}_{12},{\bf{f}}_{16}]^T$. At each time slot, we can select any beam matrix ${\bf{W}}_j\;(j\in\{1,\dots,4\})$, while the remaining beam matrices can be selected at the next time slot. Based on this, the detected signal for the $m$th analog beam at the BS can be expressed~as
\begin{eqnarray}\label{eq1}
\begin{aligned}
	y_m\!=\!&\sum_{j=1}^{M}\sum_{i=1}^{2}{\bf{v}}_m{\bf{W}}{\bf{h}}_{ji}\sqrt{P_{ji}}s_{ji}+{\bf{v}}_m{\bf{W}}{\bf{n}}_m,\\
	\!=\!&\sum_{i=1}^{2}{\bf{v}}_m{\bf{W}}{\bf{h}}_{mi}\sqrt{P_{mi}}s_{mi}\!+\!\sum_{j\neq m}^{M}\sum_{i=1}^{2}{\bf{v}}_m{\bf{W}}{\bf{h}}_{ji}\sqrt{P_{ji}}s_{ji}\!+\!{\bf{v}}_m{\bf{W}}{\bf{n}}_m,
	\end{aligned}
\end{eqnarray}
where $s_{mi}$ and $P_{mi}$, respectively, denote the transmitted signal and power at the $i$th user of the $m$th analog beam (we refer to it as User $(m,i)$), satisfying ${\mathbb{E}}\{|s_{mi}|^2\}=1$. ${\bf{W}}$ is the analog beam matrix, and we omit the subscript for simplicity. ${\bf{v}}_m\in\mathbb{C}^{1\times N}$ represents the detection vector for User $(m,i)\;(i\in\{1,2\})$. ${\bf{h}}_{mi}$ denotes the channel coefficient from User $(m,i)$ to the BS. ${\bf{n}}_m$ is independent and identically distributed (i.i.d.) additive white Gaussian noise (AWGN), and each entry is defined as  $\mathcal{CN}(0,\delta^2)$. 

For the mmWave channel, we adopt a widely used geometric channel model with $G$ scatters, where each scatter is assumed to contribute a single propagation path between the BS and user~\cite{Heath_TWC_2015}. Then, the channel ${\bf{h}}_{mi}$ can be written as
\begin{eqnarray}
{\bf{h}}_{mi}=\sqrt{\frac{N}{G}}\sum_{g=1}^G\alpha_{mi}^g{\bf{a}}(\theta_{mi}^g),
\end{eqnarray}
where $\alpha_{mi}^g$ is the complex gain of the $g$-th path with $\alpha_{mi}^g\sim\mathcal{CN}(0,{\sigma^2})$. $\theta_{mi}^g\in[0,\pi]$ is the azimuth angle of arrival  for the $g$-th path, and ${\bf{a}}(\theta_{mi}^g)$ represents the antenna array steering vector, which can  be written as
\begin{eqnarray}\label{eq5}
{\bf{a}}(\theta_{mi}^g)={\frac{1}{\sqrt{N}}}\left[1,e^{j\frac{2\pi}{\lambda}d\sin(\theta_{mi}^g)},\ldots, e^{j\frac{2\pi}{\lambda}(N-1)d\sin(\theta_{mi}^g)}\right]^T,
\end{eqnarray}
where $d$ and $\lambda$ denote the inter-antenna distance and signal wavelength, respectively.

After selecting the analog beam matrix, we define the effective channel between the BS and User $(m,i)$ as $\bar{\bf{h}}_{mi}={\bf{W}}{\bf{h}}_{mi}$. Then, (\ref{eq1}) can be rewritten as
\begin{eqnarray}\label{eq2}
\begin{aligned}
y_m=&\sum_{j=1}^{M}\sum_{i=1}^{2}{\bf{v}}_m\bar{\bf{h}}_{ji}\sqrt{P_{ji}}s_{ji}+\bar{\bf{n}}_m,\\
=&\underbrace{\sum_{i=1}^{2}{\bf{v}}_m\bar{\bf{h}}_{mi}\sqrt{P_{mi}}s_{mi}}_{\rm{Desired\;signal}}+\underbrace{\sum_{j\neq m}^{M}\sum_{i=1}^{2}{\bf{v}}_m\bar{\bf{h}}_{ji}\sqrt{P_{ji}}s_{ji}}_{\rm{Inter-beam\;interference}}+\underbrace{\bar{\bf{n}}_m}_{\rm{Noise}},
\end{aligned}
\end{eqnarray}
where $\bar{\bf{n}}_m={\bf{v}}_m{\bf{W}}{\bf{n}}_m$. 

\subsection{Proposed Advanced NOMA Decoding Scheme and EE Problem Formulation}
For the conventional NOMA decoding scheme, the SIC technique is only used within each user cluster/group~\cite{Kimy,Zeng_TVT_2019, Wang_CL_2018}.  In the uplink NOMA transmission, the BS owns global information on user signals. Therefore,  to further remove the inter-user interference, we propose a decoding scheme that only depends on the strength of the effective channel rather than the user cluster/group.  Specifically, the signal of the user owing the strongest effective channel is first detected. Before decoding User $(m,i)$'s signal, the recovered other users' (with the stronger effective channels) signals  are subtracted from the mixture signal. Accordingly, the SINR of User $(m,i)$ can be expressed~as
\begin{eqnarray}\label{sinr}
	\gamma_{mi}=\frac{|{\bf{v}}_m\bar{\bf{h}}_{mi}|^2{P_{mi}}}{{
	\sum_{jl\in\mathcal{{U}}(m,i)}|{\bf{v}}_m\bar{\bf{h}}_{jl}|^2{P_{jl}}+\delta^2}},
\end{eqnarray}
where $\mathcal{{U}}(m,i)$ denotes the users owning a  weaker effective channel than User $(m,i)$,
and its achievable rate can be written~as
\begin{eqnarray}
	R_{mi}({\bf{V}},{\bf{P}})=\log_2(1+\gamma_{mi}),
\end{eqnarray}
where ${\bf{V}}=[{\bf{v}}_1^T,\dots,{\bf{v}}_M^T]^T$ and ${\bf{P}}=[P_{mi}]_{M\times 2}$.

For User $(m,i)$, the total power consumption consists of circuit power consumption and transmit power, which can be expressed~as
\begin{eqnarray}
	P^{\rm{total}}_{mi}(P_{mi})=P_{\rm{c}}+\xi P_{mi},
\end{eqnarray}
where $P_{\rm{c}}$ denotes the circuit power consumption and $\xi>1$ denotes the inefficiency of the
power amplifier~\cite{Ng_TWC_2012}. Then, the EE of User $(m,i)$ is defined as
\begin{eqnarray}
	\eta_{mi}=\frac{R_{mi}({\bf{V}},{\bf{P}})}{P^{\rm{total}}_{mi}(P_{mi})}=\frac{\log_2(1+\gamma_{mi})}{P_{\rm{c}}+\xi P_{mi}}\;\rm{[bit/J/Hz]}.
\end{eqnarray}

In this paper, our objective is to maximize the minimum user EE subject to transmit power and rate requirements of users, which can be formulated as follows:
\begin{subequations}\label{OptA}
	\begin{align}
	\;\;\;\;\;\;\;\;\;\;&\underset{\left\{{\bf{V}, {P}}\right\}}{\rm{max}}\;\;\underset{mi}{\rm{min}}\;\;\eta_{mi}\label{OptA0}\\
	{\rm{s.t.}}\;\; &R_{mi}({\bf{V}},{\bf{P}})\geq R_{mi}^{\rm{min}},m\in \mathcal{M},i\in\{1,2\},\label{OptA1}\\
    &P_{mi}\leq P_{mi}^{\rm{max}},m\in \mathcal{M},i\in\{1,2\},\label{OptA2}\\
    &||{\bf{v}}_m{\bf{W}}||^2\leq 1,m\in \mathcal{M},\label{OptA3}
	\end{align}
\end{subequations}
where (\ref{OptA1}) denotes each user's minimum rate requirement, (\ref{OptA2}) is the maximum transmit power constraint for the users, and (\ref{OptA3}) denotes the normalized power constraint for the hybrid detection vector and analog beam at the BS.
\section{Proposed Solution}
One can observe that (\ref{OptA}) is a non-convex optimization problem, which is challenging to solve directly. In fact, we can classify (\ref{OptA}) as a generalized fractional programming~\cite{AG_1991}. To handle it, we transform (\ref{OptA0}) into a subtractive form, for which  an effective algorithm is proposed.

We define $\eta_{\rm{EE}}^{\ast}$ as the optimal EE of problem~(\ref{OptA}), and ${\bf{V}}^{\ast}$ and ${\bf{P}}^{\ast}$ are the corresponding optimal detection matrix and power allocation matrix, respectively. Then, we have 
\begin{eqnarray}
   \eta_{\rm{EE}}^{\ast}=\underset{\left\{{\bf{V}, {P}}\right\}\in{\bf{\Omega}}}{\rm{max}}\;\;\underset{mi}{\rm{min}} \frac{R_{mi}({\bf{V}},{\bf{P}})}{P^{\rm{total}}_{mi}\left(P_{mi}\right)}=\underset{mi}{\rm{min}} \frac{R_{mi}({\bf{V}}^{\ast},{\bf{P}}^{\ast})}{P^{\rm{total}}_{mi}\left(P_{mi}^{\ast}\right)},
\end{eqnarray}
where  ${{{\bf{\Omega}}}}$ is the set of all feasible solutions satisfying (\ref{OptA1})-(\ref{OptA3}). 
Regarding the optimal solution, we have the following theorem:
\begin{theorem}\label{theorem1}
The optimal solution (${\bf{V}}^{\ast}$, ${\bf{P}}^{\ast}$) of problem (\ref{OptA}) can be obtained if and only if:
\begin{eqnarray}\label{eq3}
\begin{aligned}
&\underset{\left\{{\bf{V}, {P}}\right\}\in{\bf{\Omega}}}{\rm{max}}\;\;\underset{mi}{\rm{min}} \;\;\left[{R_{mi}({\bf{V}},{\bf{P}})}-\eta_{\rm{EE}}^{\ast}{P^{\rm{total}}_{mi}\left(P_{mi}\right)}\right]\\
=&\underset{mi}{\rm{min}}\;\; \left[{R_{mi}({\bf{V}}^{\ast},{\bf{P}}^{\ast})}-\eta_{\rm{EE}}^{\ast}{P^{\rm{total}}_{mi}\left(P_{mi}^{\ast}\right)}\right]=0.
\end{aligned}
\end{eqnarray}	
\end{theorem}
\begin{proof}
The above theorem should be proved from two aspects, necessity and sufficiency. First, we look at the necessity. Assume that $\left\{{\bf{V}, {P}}\right\}$  is any feasible solution of~(\ref{eq3}), we have
\begin{eqnarray}\label{proof1}
\underset{mi}{\rm{min}} \frac{R_{mi}({\bf{V}},{\bf{P}})}{P^{\rm{total}}_{mi}\left(P_{mi}\right)}\leq \eta_{\rm{EE}}^{\ast},\;\;\;\;
\underset{mi}{\rm{min}} \frac{R_{mi}({\bf{V}}^{\ast},{\bf{P}}^{\ast})}{P^{\rm{total}}_{mi}\left(P_{mi}^{\ast}\right)}= \eta_{\rm{EE}}^{\ast}.
\end{eqnarray}
According to (\ref{proof1}), we obtain 
\begin{subequations}
\begin{align}
\underset{mi}{\rm{min}} \left\{{R_{mi}({\bf{V}},{\bf{P}})}-\eta_{\rm{EE}}^{\ast}{P^{\rm{total}}_{mi}\left(P_{mi}\right)}\right\}\leq 0,\\
\underset{mi}{\rm{min}} \left\{{R_{mi}({\bf{V}}^{\ast},{\bf{P}}^{\ast})}-\eta_{\rm{EE}}^{\ast}{P^{\rm{total}}_{mi}\left(P_{mi}^{\ast}\right)}\right\}=0.
\end{align}
\end{subequations}
Therefore, $\{{\bf{V}}^{\ast},{\bf{P}}^{\ast}\}$ is also the optimal solution of~(\ref{eq3}).  

Next, we give the proof of sufficiency. Assume that $\{{\bf{V}},{\bf{P}}\}$ and $\{{\bf{V}}^{\ast},{\bf{P}}^{\ast}\}$ are, respectively, feasible and  optimal solution of ~(\ref{eq3}), we have 
\begin{subequations}\label{proof3}
\begin{align}
\underset{mi}{\rm{min}} \left\{{R_{mi}({\bf{V}},{\bf{P}})}-\eta_{\rm{EE}}^{\ast}{P^{\rm{total}}_{mi}\left(P_{mi}\right)}\right\}\leq 0,\\
\underset{mi}{\rm{min}} \left\{{R_{mi}({\bf{V}}^{\ast},{\bf{P}}^{\ast})}-\eta_{\rm{EE}}^{\ast}{P^{\rm{total}}_{mi}\left(P_{mi}^{\ast}\right)}\right\}=0.
\end{align}
\end{subequations}
Rearranging~(\ref{proof3}) yields 
\begin{eqnarray}\label{proof4}
\underset{mi}{\rm{min}} \frac{R_{mi}({\bf{V}},{\bf{P}})}{P^{\rm{total}}_{mi}\left(P_{mi}\right)}\leq \eta_{\rm{EE}}^{\ast},\;\;\;\;
\underset{mi}{\rm{min}} \frac{R_{mi}({\bf{V}}^{\ast},{\bf{P}}^{\ast})}{P^{\rm{total}}_{mi}\left(P_{mi}^{\ast}\right)}= \eta_{\rm{EE}}^{\ast}.
\end{eqnarray}
Therefore, $\{{\bf{V}}^{\ast},{\bf{P}}^{\ast}\}$ is also the optimal solution of~(\ref{OptA}).  
\end{proof}

Theorem~\ref{theorem1} demonstrates that the solutions of problem~(\ref{OptA}) can be obtained via solving~(\ref{eq3}). However, (\ref{eq3}) is still difficult to solve since we cannot obtain $\eta_{\rm{EE}}^{\ast}$ in advance. To this end, we define the following function:
\begin{eqnarray}
	{\mathcal{L}}(\eta_{\rm{EE}})=\underset{\left\{{\bf{V}, {P}}\right\}\in{\bf{\Omega}}}{\rm{max}}\;\;\underset{mi}{\rm{min}} \;\;\left[{R_{mi}({\bf{V}},{\bf{P}})}-\eta_{\rm{EE}}{P^{\rm{total}}_{mi}\left(P_{mi}\right)}\right],
\end{eqnarray}
and we have the following  theorem:
\begin{theorem}\label{theorem2}
	${\mathcal{L}}(\eta_{\rm{EE}})$ is a strictly monotonically decreasing function with $\eta_{\rm{EE}}$.
\end{theorem}
\begin{proof}
	For any $\eta_{\rm{EE}}^1$ and $\eta_{\rm{EE}}^2$, we assume $\eta_{\rm{EE}}^1>\eta_{\rm{EE}}^2$ and (${\bf{V}}^1$, ${\bf{P}}^1$), (${\bf{V}}^2$, ${\bf{P}}^2$) as the corresponding optimal solutions. Then, we have 
\begin{equation}
\begin{aligned}
{\mathcal{L}}(\eta_{\rm{EE}}^1)=&\underset{\left\{{\bf{V}, {P}}\right\}\in{\bf{\Omega}}}{\rm{max}}\;\;\underset{mi}{\rm{min}} \;\;\left[{R_{mi}({\bf{V}},{\bf{P}})}-\eta_{\rm{EE}}^1{P^{\rm{total}}_{mi}\left(P_{mi}\right)}\right]\\
=&\underset{mi}{\rm{min}} \;\;\left[{R_{mi}({\bf{V}}^1,{\bf{P}}^1)}-\eta_{\rm{EE}}^1{P^{\rm{total}}_{mi}\left(P_{mi}^1\right)}\right]\\
<&\underset{mi}{\rm{min}} \;\;\left[{R_{mi}({\bf{V}}^1,{\bf{P}}^1)}-\eta_{\rm{EE}}^2{P^{\rm{total}}_{mi}\left(P_{mi}^1\right)}\right]\\
\leq&\underset{mi}{\rm{min}} \;\;\left[{R_{mi}({\bf{V}}^2,{\bf{P}}^2)}-\eta_{\rm{EE}}^2{P^{\rm{total}}_{mi}\left(P_{mi}^2\right)}\right]\\
=&{\mathcal{L}}(\eta_{\rm{EE}}^2),
\end{aligned}
\end{equation}
and complete the proof.
\end{proof}

Meanwhile, for a realistic system, we have ${\mathcal{L}}(\eta_{\rm{EE}})\geq 0$ with  $\eta_{\rm{EE}}=0$ and ${\mathcal{L}}(\eta_{\rm{EE}})<0$ with sufficiently large $\eta_{\rm{EE}}$. Consequently, we can use the classical bi-section method to solve ${\mathcal{L}}(\eta_{\rm{EE}})= 0$ and obtain~$\eta_{\rm{EE}}^{\ast}$, which is summarized as Algorithm~\ref{algorithm1}. 

For a given $\eta_{\rm{EE}}'$, we need to solve the following optimization problem to obtain ${\mathcal{L}}(\eta_{\rm{EE}}')$
\begin{subequations}\label{OptB}
	\begin{align}
	\;\;\;\;\;\;\;\;\;\;&\underset{\left\{{\bf{V}, {P}}\right\}}{\rm{max}}\;\;\underset{mi}{\rm{min}}\;\;{R_{mi}({\bf{V}},{\bf{P}})}-\eta_{\rm{EE}}'{P^{\rm{total}}_{mi}\left(P_{mi}\right)}\label{OptB0}\\
{\rm{s.t.}}\;\; 	&{\rm{(\ref{OptA1})-(\ref{OptA3})}}.
	\end{align}
\end{subequations}
\begin{algorithm}[t]
	{\caption{The Bi-section-Based EE Resource Allocation Algorithm.}
		\label{algorithm1}
		{\bf{Initialize}} $\eta_{\rm{EE}}^{\rm{s}}$, $\eta_{\rm{EE}}^{\rm{b}}\gg 0$ with ${\mathcal{L}}(\eta_{\rm{EE}}^{\rm{s}})\geq0$ and ${\mathcal{L}}(\eta_{\rm{EE}}^{\rm{b}})<0$, a small constant $\epsilon$.\\
		\Repeat{$|{\mathcal{L}}(\eta_{\rm{EE}}')|<\epsilon$}{Update $\eta_{\rm{EE}}'\leftarrow(\eta_{\rm{EE}}^{\rm{s}}+\eta_{\rm{EE}}^{\rm{b}})/2$,\\
			Solve problem (\ref{OptB}) and obtain ${\mathcal{L}}(\eta_{\rm{EE}}')$,\\
			$\eta_{\rm{EE}}^{\rm{s}}\leftarrow \eta_{\rm{EE}}'$ if ${\mathcal{L}}(\eta_{\rm{EE}}')>0$, else $\eta_{\rm{EE}}^{\rm{b}}\leftarrow \eta_{\rm{EE}}'$.}
		Obtain the optimal $\eta_{\rm{EE}}^{\ast}=\eta_{\rm{EE}}'$.
		
	}
\end{algorithm}
Problem (\ref{OptB}) is still difficult to solve due to the non-smooth objective function (\ref{OptB0}) and non-convex constraints (\ref{OptA1}). Next, we introduce an auxiliary variable $z$ and reformulate (\ref{OptB})~as 
\setlength{\mathindent}{0cm}
\begin{subequations}\label{OptC}
	\begin{align}
   &\;\;\;\;\;\;\;\;\;\underset{\left\{{\bf{V}, {P}},z\right\}}{\rm{max}}\;\;z\label{OptC0}\\
	&{\rm{s.t.}}\;\;{R_{mi}({\bf{V}},{\bf{P}})}-\eta_{\rm{EE}}'{P^{\rm{total}}_{mi}\left(P_{mi}\right)}\geq z, m\in \mathcal{M},i\in\{1,2\},\label{OptC1}\\ 
	 &\;\;\;\;\;{\rm{(\ref{OptA1})-(\ref{OptA3})}}.\label{OptC2}
	\end{align}
\end{subequations}

For problem (\ref{OptC}), we need to optimize three variables, i.e., $\left\{{\bf{V}, {P}},z\right\}$. One can observe from (\ref{sinr}) that $\bf{V}$ and $\bf{P}$ are coupled, and simultaneously optimizing them is challenging. Therefore, we propose an alternatively iterative optimization algorithm. Specifically, we first fix ${\bf{{P}}}$ and optimize the detection matrix ${{\bf{V}}}$ and $z$. Then, we optimize the power ${\bf{{P}}}$ and $z$ based on the obtained ${{\bf{V}}}$ at the previous iteration.
\subsection{Optimizing ${{\bf{V}}}$ and $z$ under Fixed ${\bf{{P}}}$}
Under a given feasible $\tilde{\bf{{P}}}$, we need to solve the following optimization problem:
 \begin{subequations}\label{OptD}
 	\begin{align}
 	&\;\;\;\;\;\;\;\;\;\underset{\left\{{\bf{V}},z\right\}}{\rm{max}}\;\;z\label{OptD0}\\
 	&{\rm{s.t.}}\;\;{R_{mi}({\bf{V}},\tilde{\bf{P}})}\geq z+\eta_{\rm{EE}}'{P^{\rm{total}}_{mi}\left(\tilde{P}_{mi}\right)}, m\in \mathcal{M},i\in\{1,2\},\label{OptD1}\\ 
 	&\;\;\;\;\;\;{R_{mi}({\bf{V}},\tilde{\bf{P}})}\geq R_{mi}^{\rm{min}}, m\in \mathcal{M},i\in\{1,2\},\label{OptD11}\\
 	&\;\;\;\;\;\;{\rm{(\ref{OptA3})}},\label{OptD2}
 	\end{align}
 \end{subequations}
where ${R_{mi}({\bf{V}},\tilde{\bf{P}})}=\log_2\left(1+\frac{|{\bf{v}}_m\bar{\bf{h}}_{mi}|^2{\tilde{P}_{mi}}}{{
		\sum_{jl\in\mathcal{U}(m,i)}|{\bf{v}}_m\bar{\bf{h}}_{jl}|^2\tilde{P}_{jl}+\delta^2}}\right)$. {(\ref{OptD}) is a non-convex optimization problem due to the non-convex constraint (\ref{OptD1}) and (\ref{OptD11}), and successive convex approximation technique is usually used to transform the non-convex constraint into the convex one~\cite{Rahmati_ICC_2019,Mao_2018_EUR}.} On this basis, to transform them into convex constraints, we introduce an auxiliary variable matrix ${\bf{T}}=[t_{mi}]_{M\times 2}$ and obtain
	\begin{subequations}\label{OptE}
		\begin{align}
		&\;\;\;\;\;\;\;\;\;\underset{\left\{{\bf{V}},\;{\bf{T}},\;z\right\}}{\rm{max}}\;\;z\label{OptE0}\\
		&{\rm{s.t.}}\;\;\log_2\left(1+t_{mi}\right)\geq z+\eta_{\rm{EE}}'{P^{\rm{total}}_{mi}\left(\tilde{P}_{mi}\right)}, m\in \mathcal{M},i\in\{1,2\},\label{OptE1}\\
		&\;\;\;\;\;\;\log_2\left(1+t_{mi}\right)\geq R_{mi}^{\rm{min}}, m\in \mathcal{M},i\in\{1,2\},\label{OptE11}\\
		&\;\;\;\;\;\frac{|{\bf{v}}_m\bar{\bf{h}}_{mi}|^2{\tilde{P}_{mi}}}{{
				\sum_{jl\in\mathcal{U}(m,i)}|{\bf{v}}_m\bar{\bf{h}}_{jl}|^2\tilde{P}_{jl}+\delta^2}}\geq t_{mi},m\in \mathcal{M},i\in\{1,2\},\label{OptE2} \\
		&\;\;\;\;\;\rm{(\ref{OptA3})}.\label{OptE3}
		\end{align}
	\end{subequations}

It is clear that the only non-convex constraint is (\ref{OptE2}). Next, we divide (\ref{OptE2}) into two constraints by bringing an auxiliary variable matrix  ${\bf{Q}}=[q_{mi}]_{M\times 2}$ as follows
\begin{eqnarray}
	&{|{\bf{v}}_m\bar{\bf{h}}_{mi}|^2{\tilde{P}_{mi}}}\geq t_{mi}q_{mi},\\
	&{{\sum_{jl\in\mathcal{U}(m,i)}|{\bf{v}}_m\bar{\bf{h}}_{jl}|^2\tilde{P}_{jl}+\delta^2}}\leq q_{mi}.\label{eq20}
\end{eqnarray}
Furthermore, we define $f(t_{mi},q_{mi})\triangleq t_{mi}q_{mi}$ and $u({\bf{v}}_m)\triangleq{{\bf{v}}_m\bar{\bf{H}}_{mi}{\bf{v}}_m^H{\tilde{P}_{mi}}}$, where $\bar{\bf{H}}_{mi}=\bar{\bf{h}}_{mi}\bar{\bf{h}}_{mi}^H$. After that, we linearize $u({\bf{v}}_m)$ with ${\hat{\bf{v}}}_m$, which can be expressed as 
\begin{eqnarray}
\begin{aligned}
\hat{u}({\bf{v}}_m,{\hat{\bf{v}}}_m)&\triangleq{\hat{\bf{v}}}_m\bar{\bf{H}}_{mi}{\hat{\bf{v}}}_m^H{\tilde{P}_{mi}}+\left<\nabla_{{{\bf{v}}}_m} u(\hat{\bf{v}}_m), {\bf{v}}_m-\hat{\bf{v}}_m\right>\\
&\leq u({\bf{v}}_m),
\end{aligned}
\end{eqnarray}
where $\nabla_{{{\bf{v}}}_m} u(\hat{\bf{v}}_m)$ is the derivative of $u({\bf{v}}_m)$ at $\hat{\bf{v}}_m$ and $\left<{\bf{a}},{\bf{b}}\right>\triangleq 2{\rm{Re}}\{{\bf{a}}{\bf{b}}^H\}$. In addition, we define the function $\hat{f}(t_{mi},q_{mi},\hat{t}_{mi},\hat{q}_{mi})\triangleq \frac{\hat{t}_{mi}}{2\hat{q}_{mi}}{q}_{mi}^2+\frac{\hat{q}_{mi}}{2\hat{t}_{mi}}{t}_{mi}^2$ and obtain
\begin{eqnarray}
\begin{aligned}
	&\hat{f}(t_{mi},q_{mi},\hat{t}_{mi},\hat{q}_{mi})-f(t_{mi},q_{mi})\\
	=&\frac{\hat{t}_{mi}}{2\hat{q}_{mi}}{q}_{mi}^2+\frac{\hat{q}_{mi}}{2\hat{t}_{mi}}{t}_{mi}^2-t_{mi}q_{mi}\\
	=&\frac{\hat{t}_{mi}}{2\hat{q}_{mi}}\left({q}_{mi}^2+\frac{\hat{q}_{mi}^2}{\hat{t}_{mi}^2}{t}_{mi}^2-\frac{2\hat{q}_{mi}}{\hat{t}_{mi}}t_{mi}q_{mi}\right)\\
	=&\frac{\hat{t}_{mi}}{2\hat{q}_{mi}}\left({q}_{mi}-\frac{\hat{q}_{mi}}{\hat{t}_{mi}}{t}_{mi}\right)^2\\
	\geq&0.
\end{aligned}
\end{eqnarray}

\begin{algorithm}[t]
	{\caption{The Detection Matrix Iterative Algorithm.}
		\label{algorithm2}
		{\bf{Initialize}} $\left\{\hat{\bf{V}},\;\hat{\bf{T}},\;\hat{\bf{Q}}\right\}$ and $\tilde{\bf{P}}$.\\
		\Repeat{{\rm{convergence}}}{
			Solve the optimization problem (\ref{OptH}) and obtain the optimal $\left\{{\bf{V}}^{\ast},\;{\bf{T}}^{\ast},\;{\bf{Q}}^{\ast},\;z^{\ast}\right\}$.\\
			Update $\left\{\hat{\bf{V}},\;\hat{\bf{T}},\;\hat{\bf{Q}}\right\}\leftarrow \left\{{\bf{V}}^{\ast},\;{\bf{T}}^{\ast},\;{\bf{Q}}^{\ast}\right\}$.}
	}
	Obtain the optimal $\left\{{\bf{V}}^{\ast},\;{\bf{T}}^{\ast},\;{\bf{Q}}^{\ast},\;z^{\ast}\right\}$.
\end{algorithm} 

Therefore, we have $f(t_{mi},q_{mi})\leq \hat{f}(t_{mi},q_{mi},\hat{t}_{mi},\hat{q}_{mi})$. Based on the above analysis, we can transform (\ref{OptE}) into the following optimization problem:
\begin{subequations}\label{OptH}
	\begin{align}
	&\;\;\;\;\;\;\;\;\;\underset{\left\{{\bf{V}},\;{\bf{T}},\;{\bf{Q}},\;z\right\}}{\rm{max}}\;\;z\label{OptH0}\\
	&{\rm{s.t.}}\;\;\log_2\left(1+t_{mi}\right)\geq z+\eta_{\rm{EE}}'{P^{\rm{total}}_{mi}\left(\tilde{P}_{mi}\right)}, m\in \mathcal{M},i\in\{1,2\},\label{OptH1}\\
	&\;\;\;\;\;\;t_{mi}\geq 2^{R_{mi}^{\rm{min}}}-1, m\in \mathcal{M},i\in\{1,2\},\label{OptH11}\\
	&\;\;\;\;\;{{\sum_{jl\in\mathcal{U}(j,i)}|{\bf{v}}_m\bar{\bf{h}}_{jl}|^2\tilde{P}_{jl}+\delta^2}}\leq q_{mi},m\in \mathcal{M},i\in\{1,2\},\label{OptH2} \\
	&\;\;\;\;\;\hat{u}({\bf{v}}_m,{\hat{\bf{v}}}_m)\geq\hat{f}(t_{mi},q_{mi},\hat{t}_{mi},\hat{q}_{mi}),m\in \mathcal{M},i\in\{1,2\},\label{OptH3}\\
	&\;\;\;\;\;{\rm{(\ref{OptA3})}}.\label{OptH4}
	\end{align}
\end{subequations}

{The objective function $z$ is linear.  Constraint (\ref{OptH1}) only includes a concave function  $\log_2\left(1+t_{mi}\right)$ and a linear function $z$, and thus, it is a convex constraint~\cite{CVX_2004}. In addition, (\ref{OptH11}) is a linear constraint, and (\ref{OptH2}), (\ref{OptH3}), and (\ref{OptH4}) are convex second-order cone constraints. Therefore,  (\ref{OptH}) is a convex optimization problem, which can be solved by numerical convex program solvers, e.g., interior-point method~\cite{CVX_2004}.} To obtain the solution of problem~(\ref{OptD}), we need to iteratively solve~(\ref{OptH}). Specifically, initialized from a given  feasible solution $\left\{\hat{\bf{V}},\;\hat{\bf{T}},\;\hat{\bf{Q}}\right\}$, the optimal $\left\{{\bf{V}}^{\ast},\;{\bf{T}}^{\ast},\;{\bf{Q}}^{\ast},\;z^{\ast}\right\}$ is obtained by solving (\ref{OptH}). Then, we replace $\left\{\hat{\bf{V}},\;\hat{\bf{T}},\;\hat{\bf{Q}}\right\}$ with $\left\{{\bf{V}}^{\ast},\;{\bf{T}}^{\ast},\;{\bf{Q}}^{\ast}\right\}$ and solve (\ref{OptH}) again. The above procedure is carried out until convergence. In addition, since the optimal solution $\left\{{\bf{V}}^{\ast},\;{\bf{T}}^{\ast},\;{\bf{Q}}^{\ast},\;z^{\ast}\right\}$ are obtained at each iteration, iteratively updating these variables will increase or maintain the value of the objective function~(\ref{OptD0}). Therefore, the obtained solution is at least a local optimal. We summarize the above scheme in Algorithm~\ref{algorithm2}.

 \begin{algorithm}[t]
	{\caption{The Power Iterative Algorithm.}
		\label{algorithm3}
		{\bf{Initialize}} ${\hat{\bf{P}}}$, ${\bf{{V}}}^{\ast}$.\\
		\Repeat{{\rm{convergence}}}{
			Solve the optimization problem (\ref{OptM}) and obtain the optimal $\{{{\bf{P}}}^{\ast},z^{\ast}\}$.\\
			Update ${{\hat{\bf{P}}}}\leftarrow{{\bf{P}}}^{\ast}$.}
		Obtain the optimal $\{{{\bf{P}}}^{\ast},z^{\ast}\}$.
	}
\end{algorithm} 

\subsection{Optimize ${{\bf{P}}}$ and $z$ under Fixed ${\bf{{V}}}$}
According to the obtained  ${\bf{{V}}}^{\ast}$ in Section IV. A, (\ref{OptC}) can be simplified as
 \begin{subequations}\label{OptG}
	\begin{align}
	&\;\;\;\;\;\;\;\;\;\underset{\left\{{\bf{P}},z\right\}}{\rm{max}}\;\;z\label{OptG0}\\
	&{\rm{s.t.}}\;\;{R_{mi}({\bf{{V}}}^{\ast},{\bf{P}})}\geq z+\eta_{\rm{EE}}'{P^{\rm{total}}_{mi}\left({P}_{mi}\right)}, m\in \mathcal{M},i\in\{1,2\},\label{OptG1}\\
	&\;\;\;\;\;\;{R_{mi}({\bf{{V}}}^{\ast},{\bf{P}})}\geq R_{mi}^{\rm{min}}, m\in \mathcal{M},i\in\{1,2\},\label{OptG11}\\ 
	&\;\;\;\;\;P_{mi}\leq P_{mi}^{\rm{max}},m\in \mathcal{M},i\in\{1,2\},\label{OptG2}
	\end{align}
\end{subequations}
where ${R_{mi}({\bf{{V}}}^{\ast},{\bf{P}})}=\log_2\left(1+\frac{|{\bf{{v}}}^{\ast}_m\bar{\bf{h}}_{mi}|^2{{P}_{mi}}}{{
		\sum_{jl\in\mathcal{U}(m,i)}|{\bf{{v}}}^{\ast}_m\bar{\bf{h}}_{jl}|^2{P}_{jl}+\delta^2}}\right)$. We rewrite ${R_{mi}({\bf{{V}}}^{\ast},{\bf{P}})}$ as
\begin{eqnarray}
\begin{aligned}
{R_{mi}({\bf{{V}}}^{\ast},{\bf{P}})}=&\log_2\left(\frac{{
		\sum_{jl\in\mathcal{U}(m,i)}|{\bf{{v}}}^{\ast}_m\bar{\bf{h}}_{jl}|^2{P}_{jl}}+|{\bf{{v}}}^{\ast}_m\bar{\bf{h}}_{mi}|^2{{P}_{mi}}+\delta^2}{{
		\sum_{jl\in\mathcal{U}(m,i)}|{\bf{{v}}}^{\ast}_m\bar{\bf{h}}_{jl}|^2{P}_{jl}+\delta^2}}\right)\\
	=&{R_{mi}^1({\bf{{V}}}^{\ast},{\bf{P}})}-{R_{mi}^2({\bf{{V}}}^{\ast},{\bf{P}})},
\end{aligned}
\end{eqnarray}
where ${R_{mi}^1({\bf{{V}}}^{\ast},{\bf{P}})}=\log_2\left({{
		\sum_{jl\in\mathcal{U}(m,i)}|{\bf{{v}}}^{\ast}_m\bar{\bf{h}}_{jl}|^2{P}_{jl}}\!+\!|{\bf{{v}}}^{\ast}_m\bar{\bf{h}}_{mi}|^2{{P}_{mi}}\!+\!\delta^2}\right)$, and ${R_{mi}^2({\bf{{V}}}^{\ast},{\bf{P}})}=\log_2\left({{\sum_{jl\in\mathcal{U}(m,i)}|{\bf{{v}}}^{\ast}_m\bar{\bf{h}}_{jl}|^2{P}_{jl}+\delta^2}}\right)$. To this end, constraint~(\ref{OptG1}) can be expressed as
	\begin{eqnarray}\label{DC}
	{R_{mi}^1({\bf{{V}}}^{\ast},{\bf{P}})}-{R_{mi}^2({\bf{{V}}}^{\ast},{\bf{P}})}\geq z+\eta_{\rm{EE}}'{P^{\rm{total}}_{mi}\left({P}_{mi}\right)}.
	\end{eqnarray}

Since $R_{mi}^1({\bf{{V}}}^{\ast},{\bf{P}})$ and $R_{mi}^2({\bf{{V}}}^{\ast},{\bf{P}})$ are both convex with ${\bf{P}}$, (\ref{DC}) is a difference of convex (DC) constraint~\cite{DC_1999}, and (\ref{OptG}) is a DC programming problem. In general, constrained concave convex procedure (CCCP) is  used to solve the DC program~\cite{KM_2005}. The key idea of CCCP is to transform the non-convex  set into a convex set, and then, iteratively solve the formulated convex optimization problem. The iteration is carried out until the result converges. Based on this, we first transform (\ref{DC}) into a convex constraint by the first-order Taylor approximation, which is given by 
\begin{eqnarray}
	{R_{mi}^2({\bf{{V}}}^{\ast},{\bf{P}},\hat{\bf{P}})}={R_{mi}^2({\bf{{V}}}^{\ast},\hat{\bf{P}})}+\nabla{R_{mi}^2({\bf{{V}}}^{\ast},\hat{\bf{P}})}({\bf{P}}-\hat{\bf{P}}),
\end{eqnarray}   
where ${R_{mi}^2({\bf{{V}}}^{\ast},\hat{\bf{P}})}=\log_2\left({{\sum_{jl\in\mathcal{U}(m,i)}|{\bf{{v}}}^{\ast}_m\bar{\bf{h}}_{jl}|^2\hat{P}_{jl}+\delta^2}}\right)$ and 
\begin{eqnarray}
	\nabla{R_{mi}^2({\bf{{V}}}^{\ast},\hat{\bf{P}})}({\bf{P}}-\hat{\bf{P}})=\frac{{\sum_{jl\in\mathcal{U}(m,i)}|{\bf{{v}}}^{\ast}_m\bar{\bf{h}}_{jl}|^2}(P_{jl}-\hat{P}_{jl})}{\left({{\sum_{jl\in\mathcal{U}(m,i)}|{\bf{{v}}}^{\ast}_m\bar{\bf{h}}_{jl}|^2\hat{P}_{jl}+\delta^2}}\right)\ln 2}.
\end{eqnarray}

Finally, we transform (\ref{OptG}) into 
 \begin{subequations}\label{OptM}
	\begin{align}
	&\;\;\;\;\;\;\;\;\;\underset{\left\{{\bf{P}},z\right\}}{\rm{max}}\;\;z\label{OptM0}\\
	&{\rm{s.t.}}\;\;{R_{mi}^1({\bf{{V}}}^{\ast},{\bf{P}})}-{R_{mi}^2({\bf{{V}}}^{\ast},{\bf{P}},\hat{\bf{P}})}\geq z+\eta_{\rm{EE}}'{P^{\rm{total}}_{mi}\left({P}_{mi}\right)},\label{OptM1}\\ 
	&\;\;\;\;\;|{\bf{{v}}}^{\ast}_m\bar{\bf{h}}_{mi}|^2{{P}_{mi}}\geq 2^{R_{mi}^{\rm{min}}}\Sigma_{mi},m\in \mathcal{M},i\in\{1,2\},\label{OptBM22}\\
	&\;\;\;\;\;P_{mi}\leq P_{mi}^{\rm{max}},m\in \mathcal{M},i\in\{1,2\},\label{OptBM2}
	\end{align}
\end{subequations}
where $\Sigma_{mi}={\sum_{jl\in\mathcal{U}(m,i)}|{\bf{{v}}}^{\ast}_m\bar{\bf{h}}_{jl}|^2{P}_{jl}+\delta^2}$.  Problem (\ref{OptM}) is a standard convex optimization problem and can be solved by the interior-point method. Likewise, we need to iteratively solve (\ref{OptM}) to obtain the solution of (\ref{OptG}). Specifically, starting with an initial feasible $\hat{\bf{P}}$, the optimal ${\bf{P}}^{\ast}$ can be obtained via solving (\ref{OptM}). Then, we update $\hat{\bf{P}}$ with ${\bf{P}}^{\ast}$ and resolve (\ref{OptM}). The above iteration is carried out until convergence. We summarize the above scheme in Algorithm~\ref{algorithm3}.

\begin{figure}[t]
	\begin{center}
		\includegraphics[width=9cm,height=7cm]{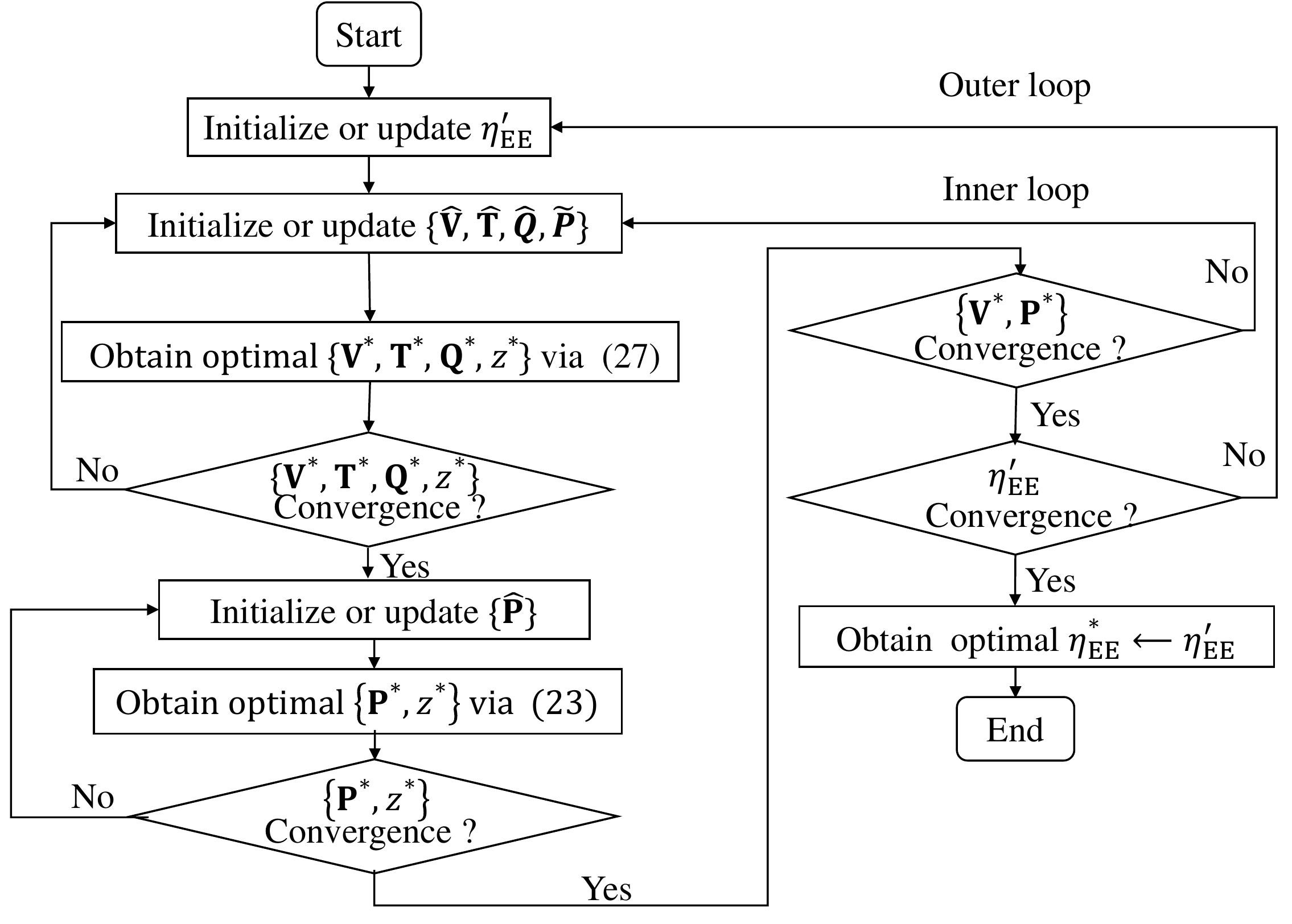}
		\caption{The flow chart of the proposed two-loop iterative algorithm.}
		\label{figure3}
	\end{center}
\end{figure}

\begin{remark}
We show the algorithm flow chart for solving the original problem~(\ref{OptA}) in Fig.~\ref{figure3}, which includes the inner and outer loops. In the inner loop, we need to solve the optimization problem~(\ref{OptB}), and an alternatively iterative algorithm is proposed. Since (\ref{OptH}) and (\ref{OptM}) are standard convex optimization problems, the obtained solution $\left\{{\bf{V}}^{\ast},\;{\bf{T}}^{\ast},\;{\bf{Q}}^{\ast},\;{\bf{P}}^{\ast},\;z^{\ast}\right\}$ is optimal at each iteration. Therefore, iteratively updating those variables will always increase or at least maintain the objective value of (\ref{OptB})~\cite{Zhang_TVT_2016}. On the other hand, the objective value of (\ref{OptB}) has an upper bound due to the limited transmit power. Thus, the proposed inner-loop iterative algorithm will converge to a stationary and at least a local optimal solution for problem~(\ref{OptB}).  Next, the outer loop is needed to solve problem~(\ref{OptA}), and the bi-section-based iterative algorithm is adopted to obtain the optimal solution of problem~(\ref{OptA}).  
\end{remark}

\subsection{Analysis of Computational Complexity}
Now, we analyze the computational complexity of the proposed two-loop iterative algorithm for solving problem (\ref{OptA}). For the inner iteration, we need to iteratively solve (\ref{OptH}) and (\ref{OptM}). The computational complexity of solving (\ref{OptH}) is $\mathcal{O}([6M+1]^{3.5})$~\cite{Huang_SPL_2015}, where $6M+1$ denotes the number of variables. In addition, the computational complexity of the CCCP-based iterative algorithm to solve~(\ref{OptM}) is  $\mathcal{O}(\log((6M)/\varepsilon^{\rm{o}}\zeta)/\log(\xi))$ at each iteration, where $6M$ is the total number of constraints in problem (\ref{OptM}). $\varepsilon^{\rm{o}}$ is the initial point for approximating the accuracy, $\zeta(\zeta\!\in\!(0,1])$ is the stopping criterion, and $\xi$ is used for updating the accuracy~\cite{Mokari_TVT_2016}. We assume that $L$ inner iterations are needed, and the computational complexity of solving (\ref{OptC}) is $\mathcal{O}(L([6M+1]^{3.5}+\log((6M)/\varepsilon^{\rm{o}}\zeta)/\log(\xi)))$. Since the computational complexity of the outer iteration is $\mathcal{O}(\log_2(1/\epsilon))$, the total computational complexity of our proposed algorithm is $\mathcal{O}(\log_2(1/\epsilon)L([6M+1]^{3.5}+\log((6M)/\varepsilon^{\rm{o}}\zeta)/\log(\xi)))$.

\section{ZF-Based Low Complexity Algorithm}
In Section III, we proposed an alternatively iterative ${\bf{V}}$ and ${\bf{P}}$ algorithm for solving (\ref{OptC}). To decrease the computational complexity,  in this section, we develop a ZF-based algorithm.  First, we arrange the two users in each beam group following a descending order based on their effective channel strengths, i.e., $|\bar{\bf{h}}_{m1}|\geq|\bar{\bf{h}}_{m2}|$. Similar to~\cite{Kimy}, we generate the detection vectors based on the effective channel of the strong users.  Thus, we define ${\bf{H}}=[{\bar{\bf{h}}}_{1,1},\dots,{\bar{\bf{h}}}_{M,1}]$, and ${\bf{V}}=({\bf{H}}^H{\bf{H}})^{-1}{\bf{H}}^H$. The detection vector ${\bf{v}}_{m}$ can be expressed as
\begin{eqnarray}
	{\bf{v}}_{m}=\frac{{\bf{V}}(m)}{||{\bf{V}}(m){\bf{W}}\|}, m\in\mathcal{M},
\end{eqnarray}
where ${\bf{W}}(m)$ denote the $m$th row of ${\bf{W}}$. 

{After designing the detection matrix {\bf{V}}, the interference among strong users of all clusters can be canceled, and the detected signal of the $m$-th beam at the BS can be expressed as}
\begin{eqnarray}\label{ZF1}
y_{m}=\sum_{i=1}^{2}{\bf{v}}_m\bar{\bf{h}}_{mi}\sqrt{P_{mi}}s_{mi}+\sum\nolimits_{j=1}^{M}{\bf{v}}_m\bar{\bf{h}}_{j2}\sqrt{P_{j2}}s_{j2}+\bar{\bf{n}}_{mi}.
\end{eqnarray}

\begin{figure}[t]
	\begin{center}
		\includegraphics[width=5cm,height=7cm]{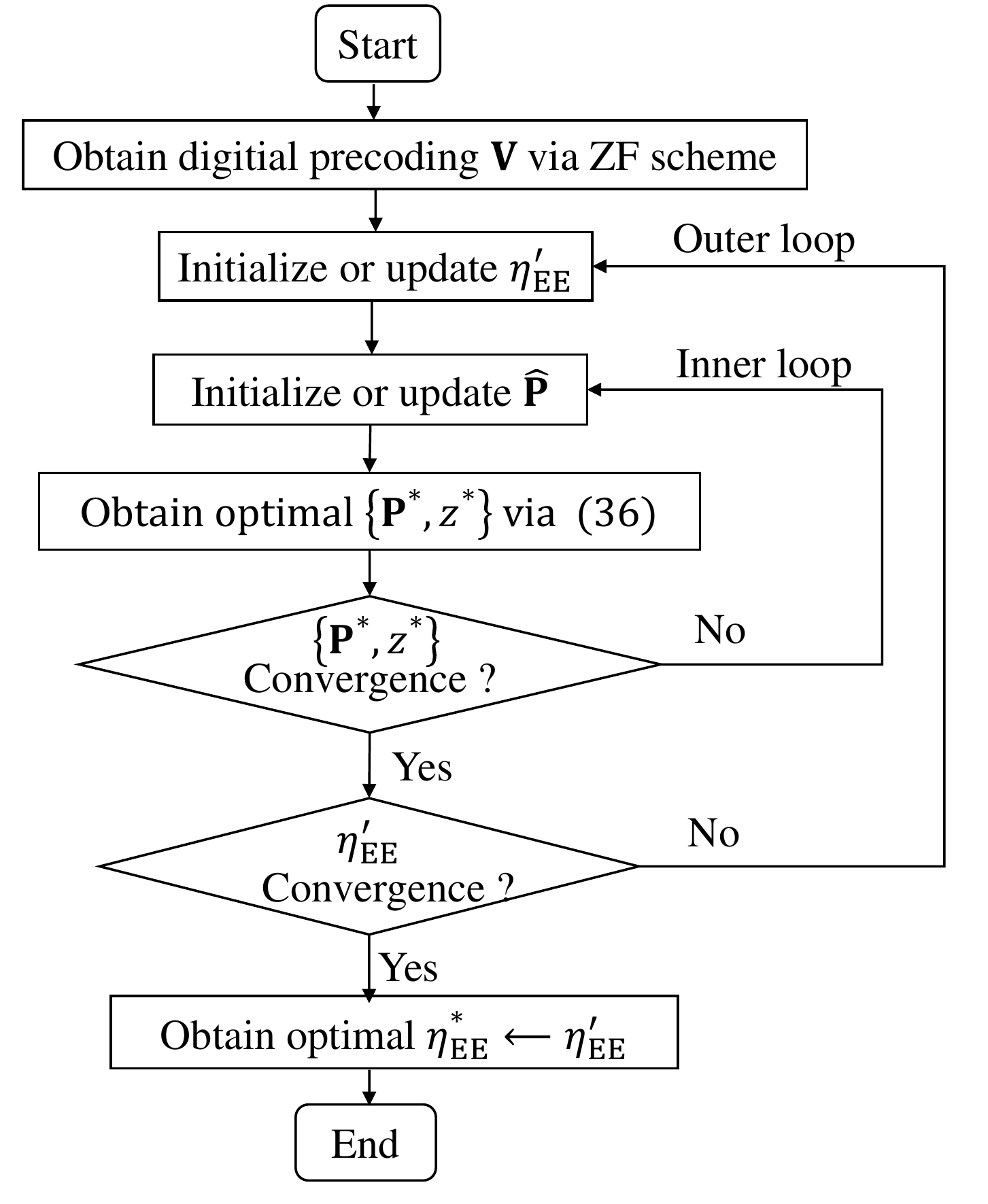}
		\caption{The flow chart of the proposed ZF-based two-loop iterative  algorithm.}
		\label{figure4}
	\end{center}
\end{figure}
In addition, we adopt the same decoding order as the one used in Section III. Similarly, we define ${\hat{\mathcal{U}}}_2(m,i)$ as the second user set in all clusters that owns weaker effective channel than User $(m,i)$. As a result, the SINR of User $(m,i)$ can be calculated as
\begin{eqnarray}\label{ZF}
	\gamma_{mi}^{\rm{zf}}=\frac{|{\bf{v}}_m\bar{\bf{h}}_{mi}|^2{P_{mi}}}{{
			\sum_{jl\in{\hat{\mathcal{{U}}}}_2(m,i)}|{\bf{v}}_m\bar{\bf{h}}_{jl}|^2{P_{jl}}+\delta^2}},
\end{eqnarray}
and the achievable rate can be expressed as
\begin{eqnarray}
		R_{mi}^{\rm{zf}}({\bf{P}})=\log_2(1+\gamma_{mi}^{\rm{zf}}).
\end{eqnarray}

Then, we reformulate the following max-min EE optimization problem as
 \begin{subequations}\label{OptW}
 	\begin{align}
 	\;\;\;\;\;\;\;\;&\underset{\left\{\bf{{P}}\right\}}{\rm{max}}\;\;\underset{mi}{\rm{min}}\;\;\eta_{mi}\label{OptW0}\\
 	&{\rm{s.t.}}\;\; R_{mi}^{\rm{zf}}({\bf{P}})\geq R_{mi}^{\rm{min}},m\in \mathcal{M},i\in\{1,2\},\label{OptW1}\\
 	&\;\;\;\;\;P_{mi}\leq P_{mi}^{\rm{max}},m\in \mathcal{M},i\in\{1,2\}.\label{OptW2}
 	\end{align}
 \end{subequations}

According to the scheme proposed in Section III, we directly transform (\ref{OptW}) into the following:
\setlength{\mathindent}{0cm}
\begin{subequations}\label{OptY}
	\begin{align}
	\;\;\;\;\;\;\;&\;\;\;\;\;\;\;\;\;\;\;\underset{\left\{{ \bf{P}},z\right\}}{\rm{max}}\;\;z\label{OptY0}\\
	&{\rm{s.t.}}\;\;{R_{mi}^{\rm{zf}}({\bf{P}})}\geq z+\eta_{\rm{EE}}'{P^{\rm{total}}_{mi}\left(P_{mi}\right)}, m\in \mathcal{M},i\in\{1,2\},\label{OptY1}\\ 
	&\;\;\;\;\;{\rm{(\ref{OptW1}), (\ref{OptW2})}}.\label{OptY4}
	\end{align}
\end{subequations}

Finally, we can adopt the same scheme proposed in Section IV.~B to solve the above problem. Compared with the first scheme, the proposed ZF scheme only needs to update P iteratively, as $\bf{V}$ is fixed based on~(\ref{ZF}).  In addition, one can easily obtain that the computational complexity of the ZF-based two-loop iterative algorithm is $\mathcal{O}(\log_2(1/\epsilon)L\log((6M)/\varepsilon^{\rm{o}}\zeta)/\log(\xi))$, which is much lower than that of the proposed alternatively iterative algorithm in Section IV. We provide the algorithm flow chart in Fig.~\ref{figure4}.

\section{Simulation Results}

In this section, we evaluate the performance of our proposed algorithms for the MIMO-NOMA mmWave system.  The default simulation parameters are set as follows: The BS is equipped with $N=32$ antennas and $M=4$ RF chains. We assume that there are enough users to form multiple two-user beam groups. The number of clusters in the mmWave channel is assumed $G=3$, $\alpha_{mi}^{g}\sim\mathcal{CN}(0,1)$ and $\theta_{mi}^{g}$ follows the uniform distribution at $[-\pi,\pi]$. Meanwhile, we define the signal-to-noise ratio (${\rm{SNR}}$) as ${\rm{SNR}}=P_{mi}^{\rm{max}}/\delta^2$, and assume that all users have the same maximum transmit power. The inefficiency of the power amplifier $\xi$ is set as $1/0.38$, while the circuit power consumption of each user is set as $P_{\rm{c}}=100$ mW. The minimum rate requirement is assumed the same for all users and set as $R_{mi}^{\rm{min}}=0.2$ bps/Hz.  For the sake of analysis, we refer to the jointly iterative ${\bf{{P}}}$ and  ${{\bf{V}}}$ algorithm as Scheme 1, and the ZF-based iterative ${\bf{{P}}}$ algorithm as Scheme 2.

\begin{figure}[t]
	\begin{center}
		\includegraphics[width=9cm,height=8cm]{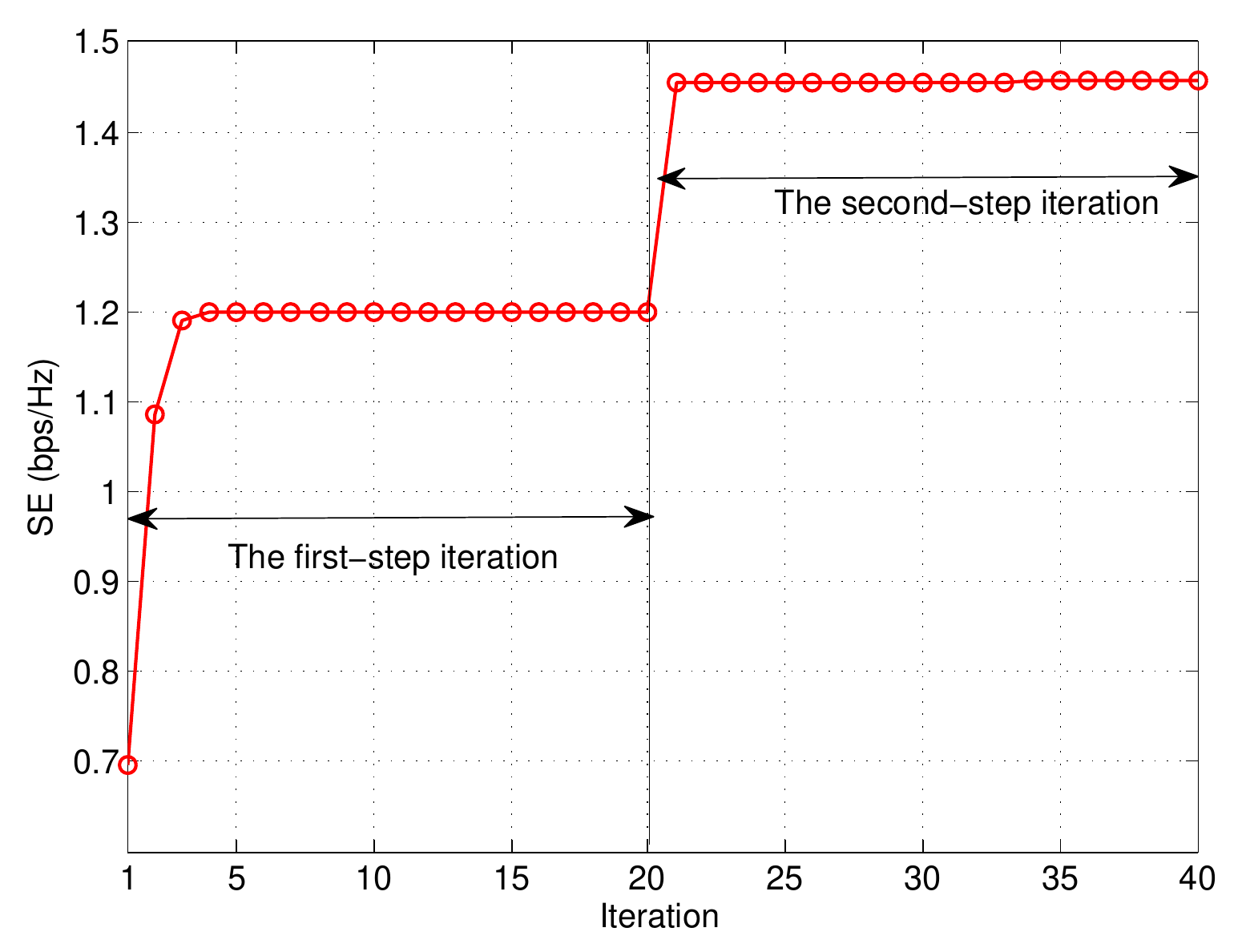}
		\caption{SE versus iteration (alternatively iterative ${\bf{{P}}}$ and  ${{\bf{V}}}$).}
		\label{Sfigure1}
	\end{center}
\end{figure}
\begin{figure}[t]
	\begin{center}
		\includegraphics[width=9cm,height=8cm]{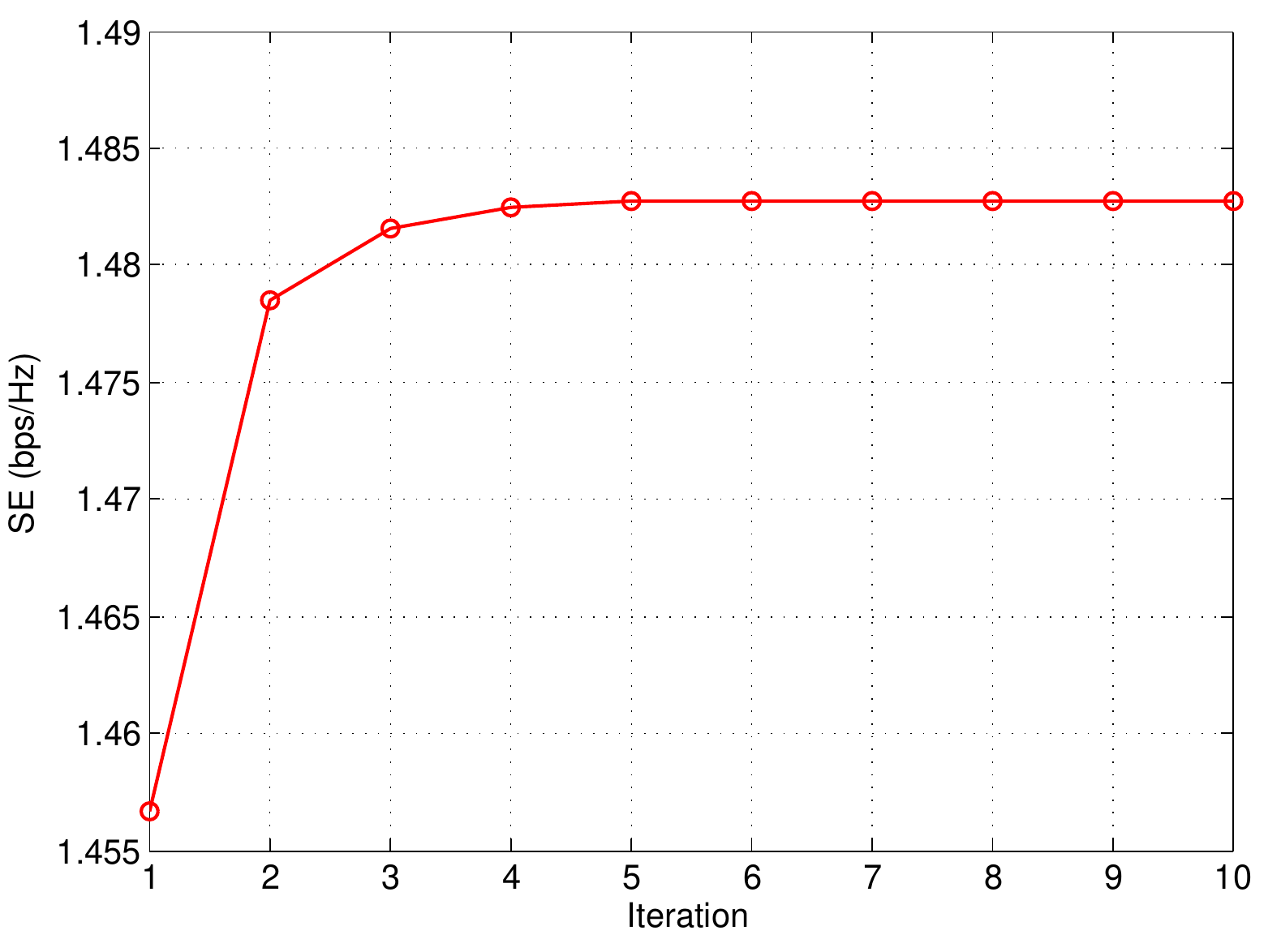}
		\caption{SE versus inner iteration.}
		\label{Sfigure2}
	\end{center}
\end{figure}
\begin{figure}[t]
	\begin{center}
		\includegraphics[width=9cm,height=8cm]{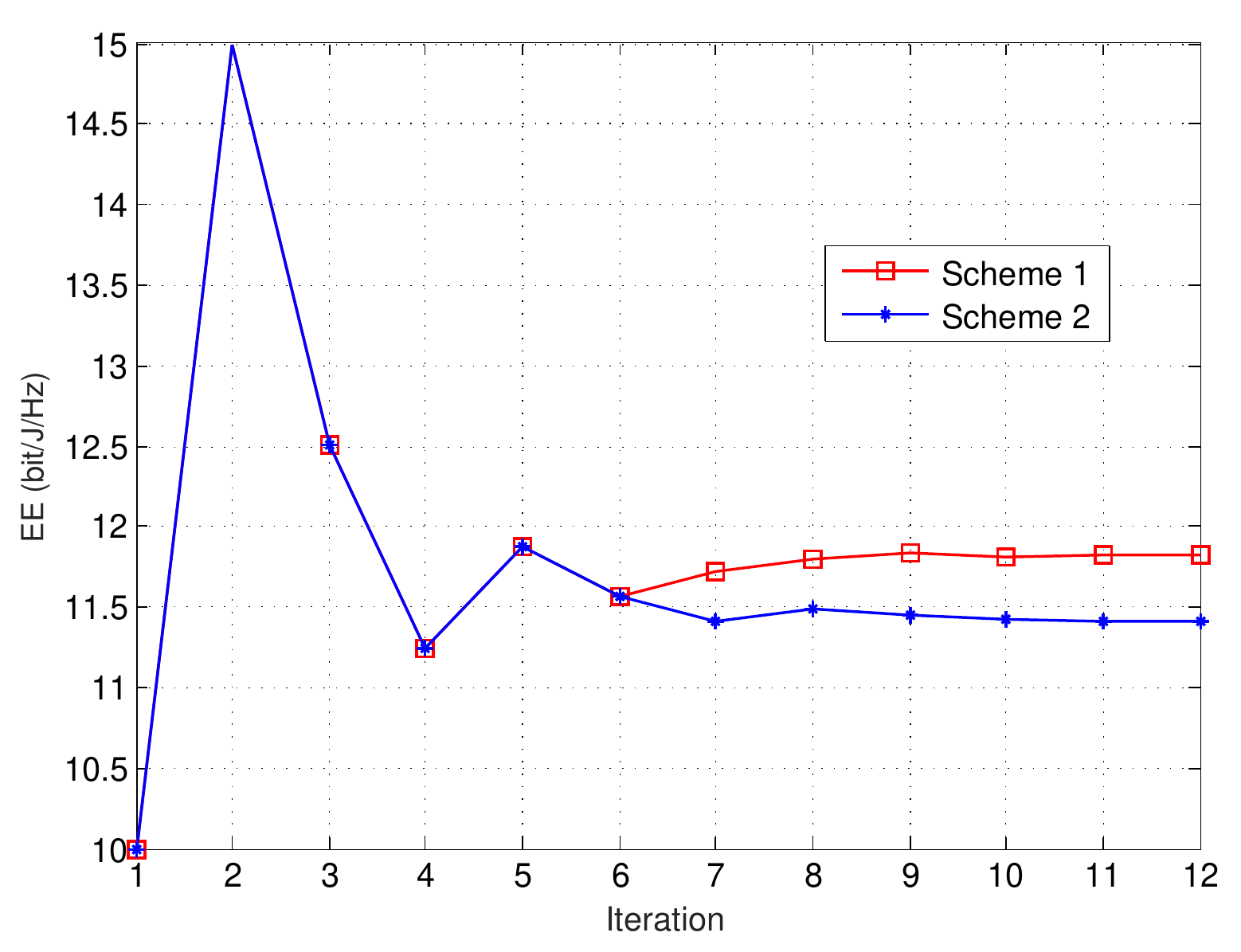}
		\caption{EE versus outer iteration.}
		\label{Sfigure3}
	\end{center}
\end{figure}
\begin{figure}[t]
	\begin{center}
		\includegraphics[width=9cm,height=8cm]{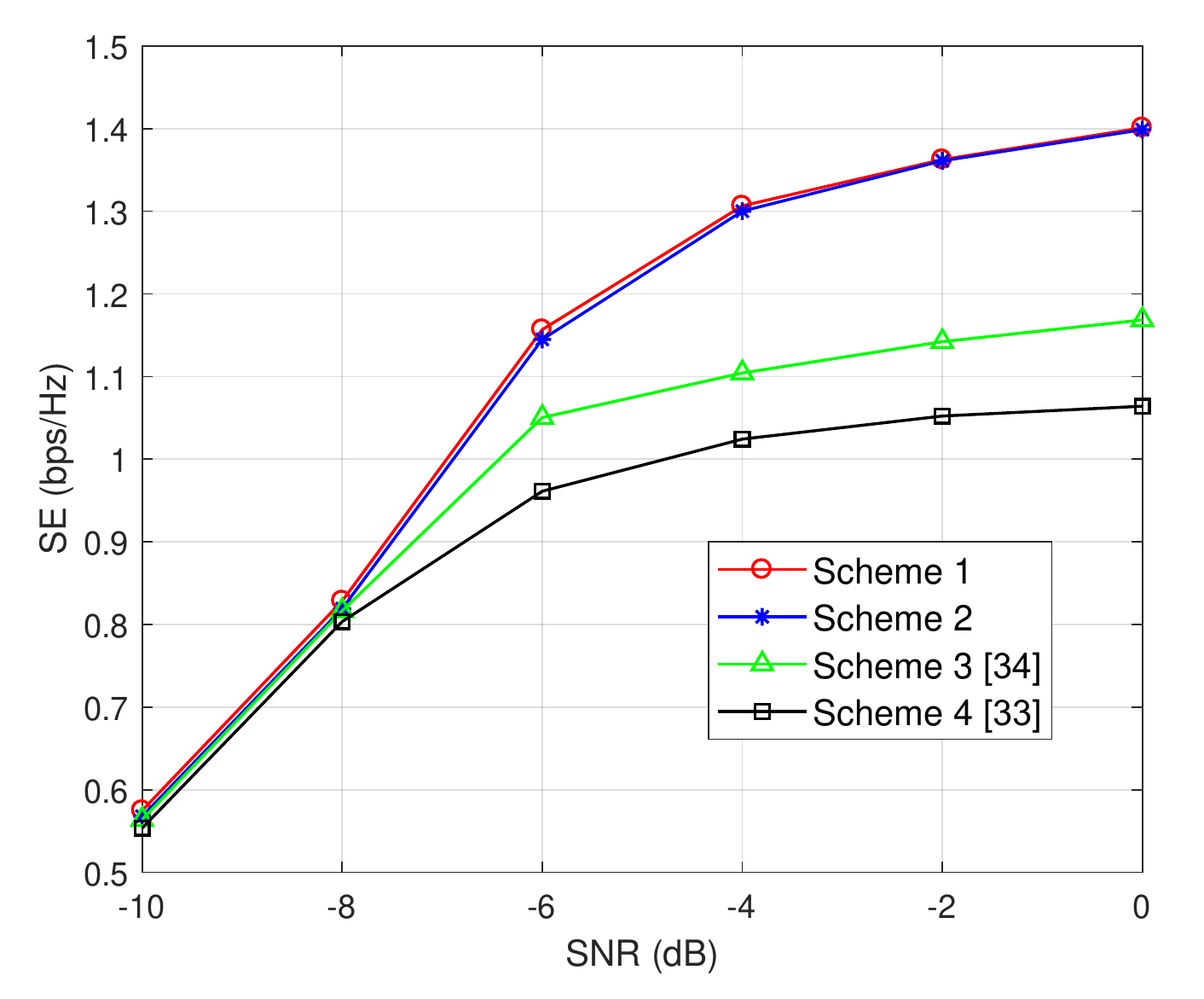}
		\caption{SE versus SNR for different schemes.}
		\label{Sfigure4}
	\end{center}
\end{figure}
\begin{figure}[t]
	\begin{center}
		\includegraphics[width=9cm,height=8cm]{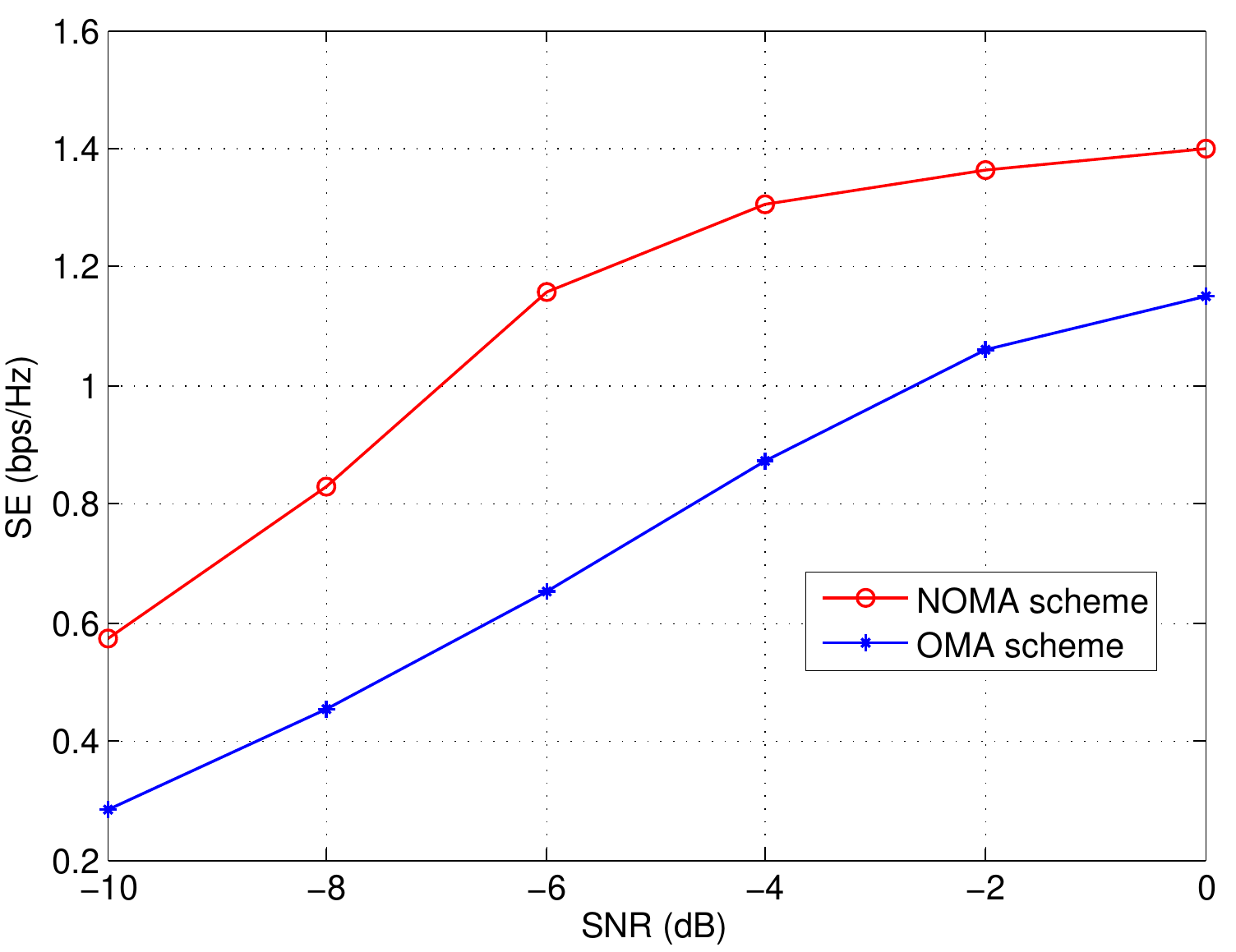}
		\caption{SE versus SNR for NOMA and OMA.}
		\label{Sfigure5}
	\end{center}
\end{figure}
\begin{figure}[t]
	\begin{center}
		\includegraphics[width=9cm,height=8cm]{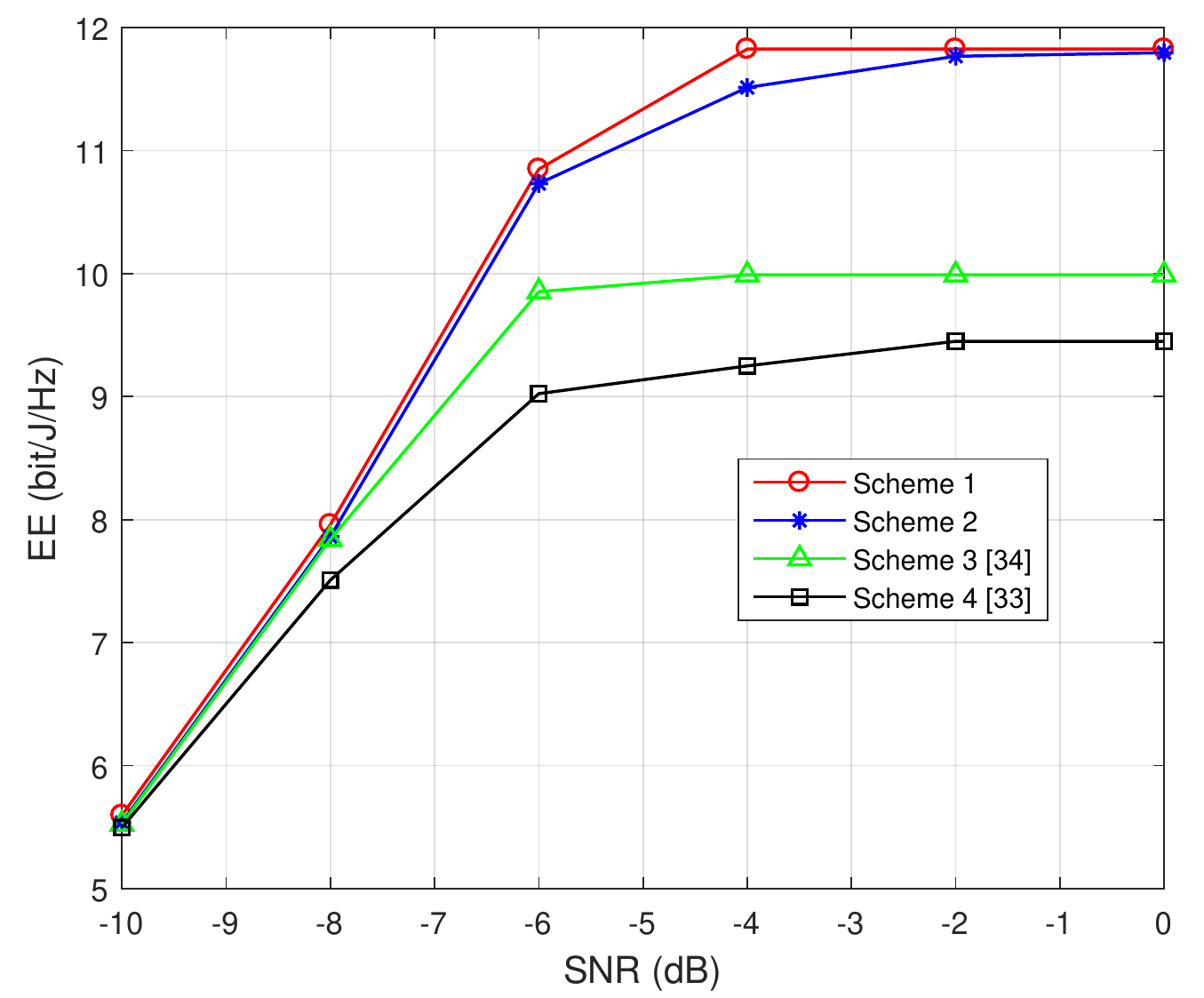}
		\caption{EE versus SNR for different schemes.}
		\label{Sfigure6}
	\end{center}
\end{figure}
\begin{figure}[t]
	\begin{center}
		\includegraphics[width=9cm,height=8cm]{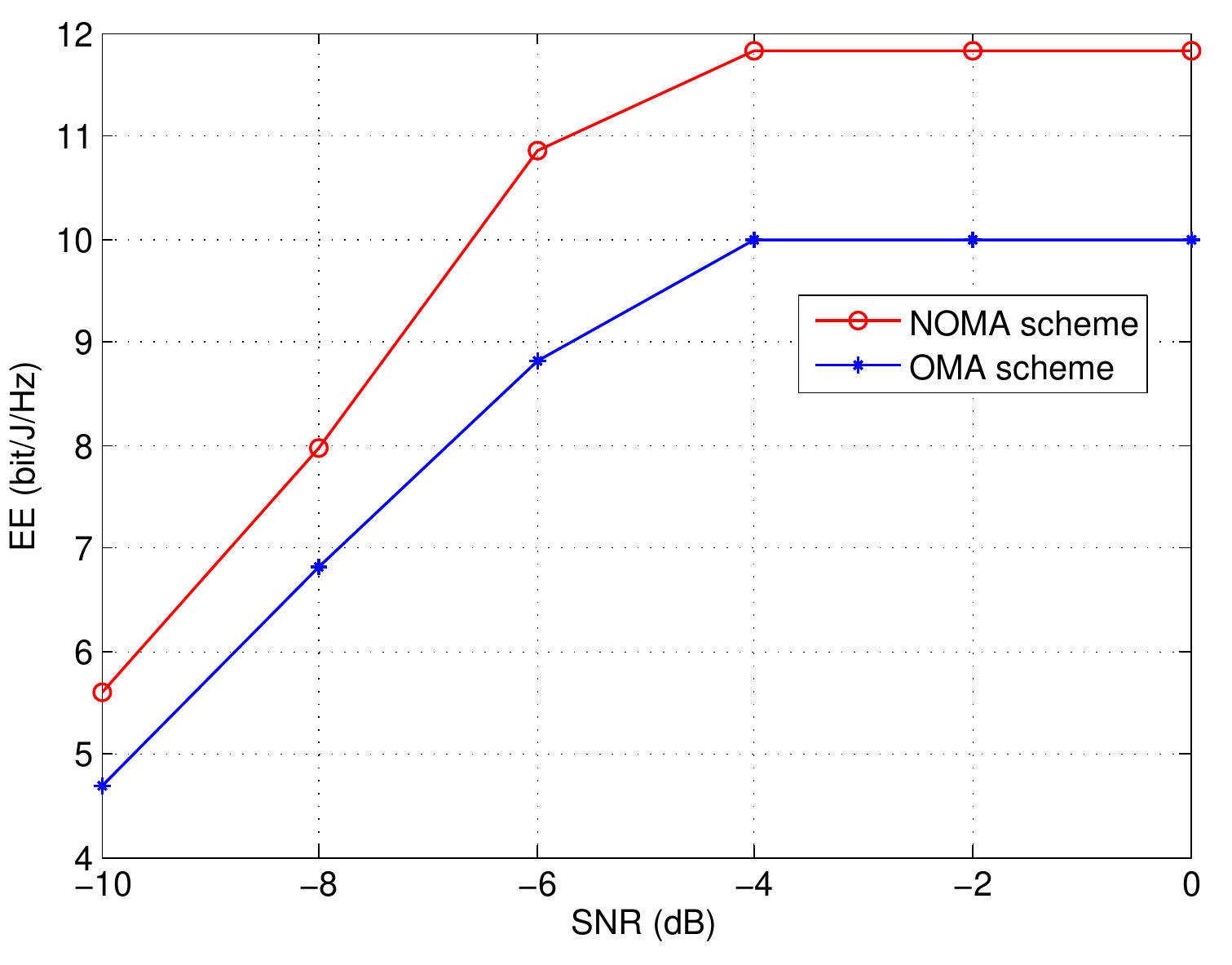}
		\caption{EE versus SNR for NOMA and OMA.}
		\label{Sfigure7}
	\end{center}
\end{figure}

To show the convergence performance of Scheme~1,  we first set $\eta_{\rm{EE}}'=0$ and plot the SE versus the iteration number (including ${\bf{{P}}}$ and  ${{\bf{V}}}$)  in Fig.~\ref{Sfigure1}. Here, ``The first-step iteration" stands for iteratively updating ${{\bf{V}}}$, namely solving problem~(\ref{OptD}), and ``The second-step iteration" stands for iteratively updating ${\bf{{P}}}$, i.e., solving problem~(\ref{OptG}). One can observe that these two steps converge fast. For example, 4 iterations are needed for the first step to converge, while only  2 iterations are needed for the second step  to converge. For Scheme~2, it is obvious that we only need to iteratively update ${{\bf{P}}}$, namely ``The second-step iteration."  In addition, Fig.~\ref{Sfigure2} shows the SE versus the number of inner iterations  when $\eta_{\rm{EE}}'=0$, i.e., solving problem~(\ref{OptC}). We find that 4 iterations are needed for convergence.

Fig.~\ref{Sfigure3} plots the EE versus the number of outer iterations for the proposed two schemes. As expected, the EE fluctuates due to the adopted bi-section method. Moreover, we find that Schemes 1 and 2 almost simultaneously converge after about 8 iterations. Note that although they have the same iteration number, Scheme 2 does not need the ``first-step iteration" presented in Fig.~\ref{Sfigure1}. Therefore, Scheme 2 can obtain a stable solution faster in comparison with Scheme 1. However,  it can be observed that the obtained EE in Scheme~2 is lower than that in Scheme~1. 

{Fig.~\ref{Sfigure4} shows the SE versus SNR for different schemes, where  we set $\eta_{\rm{EE}}'=0$. In addition to the proposed Schemes 1 and 2, we also show the SE results provided by the conventional uplink decoding order which depends on the  strength of the effective channel and the user group or cluster~\cite{Wang_CL_2018} (Scheme 3), and another baseline scheme (Scheme 4), where the weak-interference is not removed among clusters~\cite{Kimy}. Among all considered schemes, Scheme 1 is always the best, followed by Schemes 2, 3, and 4. The gap between the proposed schemes (Schemes 1 and 2) and Scheme 3 illustrates the effectiveness of removing the interference based on the strength, especially in high SNR. Meanwhile, the gap between Scheme 3 and Scheme 4 shows the necessity of removing the weak-interference among the clusters.}  Furthermore, in Fig.~\ref{Sfigure5} we compare the SE of NOMA (i.e., Scheme 1) with  that of the conventional OMA scheme, where users belonging to the same beam group are served by time duplex division access. It is obvious that the proposed NOMA scheme can obtain a higher SE than the conventional OMA one. 

{Fig.~\ref{Sfigure6}  plots the EE versus SNR for the above four schemes. For all four schemes, the EE first increases and then saturates as the SNR increases. In low SNR regime, the small increase of the SNR can yield a large increase in SE (as shown in~Fig.~\ref{Sfigure4}), and thus, a large increase in EE. In contrast, in high SNR regime, a large increase of the SNR only leads to a small increase in SE. As a result, the extra available power may not be used for increasing the EE. As in Fig.~\ref{Sfigure4}, the proposed schemes always outperform the baselines, especially for high SNR.} Finally, the EE for  NOMA (i.e., Scheme 1) and OMA scheme are compared in Fig.~\ref{Sfigure7}. One can observe that our proposed scheme has a higher EE when compared with OMA.
\section{Conclusion}
In this paper, we have investigated the EE problem in an uplink MIMO-NOMA mmWave system. We formulated a max-min EE optimization problem involving a joint optimization of the transmit power at the users and detection matrix at the BS. We proposed two schemes to solve this problem. Simulation results confirmed that the proposed NOMA schemes outperform other NOMA baseline algorithms, as well as the conventional OMA scheme in terms of SE and EE.

\bibliographystyle{ieeetr}
\bibliography{references}
\end{document}